\pdfoutput=1
\documentclass{article}

\usepackage{graphicx} 
\usepackage[caption=false]{subfig}
\usepackage{comment}

\RequirePackage[OT1]{fontenc}

\usepackage[colorlinks=true, citecolor=blue, filecolor=black, linkcolor=red, urlcolor=blue, hypertexnames=false, pdftex]{hyperref}

\usepackage[ansinew]{inputenc}
\usepackage{lmodern}

\usepackage{booktabs} 
\newcommand{\ra}[1]{\renewcommand{\arraystretch}{#1}}

\usepackage{amsmath}
\usepackage{amsthm}
\usepackage{amsfonts}
\usepackage{amssymb}
\usepackage{url}
\usepackage{mhequ}
\usepackage{authblk}

\usepackage{algorithm} 
\usepackage{algpseudocode}
\usepackage{setspace}
\usepackage{todonotes}

\usepackage{ifpdf}
\ifpdf
  \usepackage{epstopdf} 
\fi

\usepackage[letterpaper]{geometry}
\DeclareMathOperator*{\argmax}{arg\,max}

\makeatletter
\newcommand{\StatexIndent}[1][3]{%
  \setlength\@tempdima{\algorithmicindent}%
  \Statex\hskip\dimexpr#1\@tempdima\relax}
	
\algdef{S}[WHILE]{WhileNoDo}[1]{\algorithmicwhile\ #1}%
\makeatother

\newcommand{\etal}{\textit{et al.\ }}

\theoremstyle{plain} 
\newtheorem{theorem}{Theorem}[section]
\newtheorem*{theorem*}{Theorem}
\newtheorem{lemma}[theorem]{Lemma}
\newtheorem*{lemma*}{Lemma}

\newtheorem*{corollary*}{Corollary}

\newtheorem*{proposition*}{Proposition}

\newtheorem*{definition*}{Definition}
\newtheorem{example}[theorem]{Example}
\newtheorem*{example*}{Example}
\newtheorem{remark}[theorem]{Remark}

\newtheorem*{remark*}{Remark}
\newtheorem*{remarks*}{Remarks}

\newtheorem{ass}[theorem]{Assumption}

\def \d{\mathbf{d}}
\def \be{\begin{equs}}
\def \ee{\end{equs}}
\def \P{\mathbb{P}}
\def \E{\mathbb{E}}

\newcommand \TV{\mathrm{TV}}

\def \defi{\mathrm{def}}
\graphicspath{{figures/}}
\usepackage{natbib}
\begin{document}

\title{Accelerating Asymptotically Exact MCMC for Computationally Intensive Models via Local Approximations}

%
\author[1]{Patrick R.\ Conrad}
\author[1]{Youssef M.\ Marzouk}
\author[2]{Natesh S.\ Pillai}
\author[3]{Aaron Smith}
\affil[1]{Department of Aeronautics and Astronautics\\Massachusetts Institute of Technology \authorcr 77 Massachusetts Avenue\\Cambridge, MA 02139, USA. \texttt{\{prconrad, ymarz\}@mit.edu}}
\affil[2]{Department of Statistics\\Harvard University\authorcr 1 Oxford Street\\Cambridge, MA 02138, USA. \texttt{pillai@stat.harvard.edu}}
\affil[3]{Department of Mathematics and Statistics\\University of Ottawa \authorcr 585 King Edward Avenue\\Ottawa, ON K1N 7N5, Canada. \texttt{asmi28@uottawa.ca}}

\maketitle






\begin{abstract}
We construct a new framework for accelerating Markov chain Monte Carlo in posterior sampling problems where standard methods are limited by the computational cost of the likelihood, or of numerical models embedded therein. Our approach introduces local approximations of these models into the Metropolis-Hastings kernel, borrowing ideas from deterministic approximation theory, optimization, and experimental design. Previous efforts at integrating approximate models into inference typically sacrifice either the sampler's exactness or efficiency; our work seeks to address these limitations by exploiting useful convergence characteristics of local approximations. We prove the ergodicity of our approximate Markov chain, showing that it samples asymptotically from the \emph{exact} posterior distribution of interest. We describe variations of the algorithm that employ either local polynomial approximations or local Gaussian process regressors. Our theoretical results reinforce the key observation underlying this paper: when the likelihood has some local regularity, the number of model evaluations per MCMC step can be greatly reduced without biasing the Monte Carlo average. Numerical experiments demonstrate multiple order-of-magnitude reductions in the number of forward model evaluations used in representative ODE and PDE inference problems, with both synthetic and real data.
\end{abstract}



\section{Introduction}
Bayesian inference for computationally intensive models is often limited by the computational cost of Markov chain Monte Carlo (MCMC) sampling. For example, scientific models in diverse fields such as geophysics, chemical kinetics, and biology often invoke ordinary or partial differential equations to describe the underlying physical or natural phenomena. These differential equations constitute the \emph{forward model} which, combined with measurement or model error, yield a likelihood function. Given a numerical implementation of this physical model, standard MCMC techniques are in principle appropriate for sampling from the posterior distribution. However, the cost of running the forward model anew at each MCMC step can quickly become prohibitive if the forward model is computationally expensive. 

An important strategy for mitigating this cost is to recognize that the forward model may exhibit regularity in its dependence on the parameters of interest, such that the model outputs may be approximated with fewer samples than are needed to characterize the posterior via MCMC. Replacing the forward model with an approximation or ``surrogate'' \textit{decouples} the required number of forward model evaluations from the length of the MCMC chain, and thus can vastly reduce the overall cost of inference~\citep{Sacks1989, Kennedy2001}. Existing approaches typically create high-order global approximations for either the forward model outputs or the log-likelihood function using, for example, global polynomials \citep{Marzouk2007,Marzouk2009a}, radial basis functions \citep{Bliznyuk2012, Joseph2012}, or Gaussian processes \citep{Sacks1989, Kennedy2001, Rasmussen2003, Santner2003}. As in most of these efforts, we will assume that the forward model is deterministic and available only as a black box, thus limiting ourselves to ``non-intrusive'' approximation methods that are based on evaluations of the forward model at selected input points.\footnote{Interesting examples of intrusive techniques exploit multiple spatial resolutions of the forward model~\citep{higdon2003, christen2005, Efendiev2006}, models with tunable accuracy~\citep{Korattikara2013, Bal2013}, or projection-based reduced order models \citep{Frangos2010, Lieberman2010, Cui2014}.} 
Since we assume that the exact forward model is available and computable, but simply too expensive to be run a large number of times, the present setting is distinct from that of either pseudo-marginal MCMC or approximate Bayesian computation (ABC); these are important methods for intractable posteriors where the likelihood can only be estimated or simulated from, respectively \citep{Andrieu2009, Marin2011}.\footnote{Typically the computational model itself is an approximation of some underlying governing equations. Though numerical discretization error can certainly affect the posterior \citep{Kaipio2007}, we do not address this issue here; we let a numerical implementation of the forward model, embedded appropriately in the likelihood function, define the exact posterior of interest.}

Although current approximation methods can provide significant empirical performance improvements, they tend either to over- or under-utilize the surrogate, sacrificing exact sampling or potential speedup, respectively. In the first case, many methods produce some fixed approximation, inducing an approximate posterior. In principle, one might require only that the bias of a posterior expectation computed using samples from this approximate posterior be small relative to the variance introduced by the finite length of the MCMC chain, but current methods lack a rigorous approach to controlling this bias \citep{Bliznyuk2008, Fielding2011}; \cite{Cotter2010} show that bounding the bias is in principle possible, by proving that the rate of convergence of the forward model approximation can be transferred to the approximate posterior, but their bounds include unknown constants and hence do not suggest practical strategies for error control. Conversely, other methods limit potential performance improvement by failing to ``trust'' the surrogate even when it is accurate. Delayed-acceptance schemes, for example, eliminate the need for error analysis of the surrogate but require at least one full model evaluation for each accepted sample \citep{Rasmussen2003, christen2005, cui2011wrr}, which remains a significant computational effort.

Also, analyzing the error of a forward model approximation can be quite challenging for the \textit{global} approximation methods used in previous work---in particular for methods that use complex sequential experimental design heuristics to build surrogates over the posterior \citep{Rasmussen2003, Bliznyuk2008, Fielding2011}. 
Even when these design heuristics perform well, it is not clear how to establish rigorous error bounds for finite samples or even how to establish convergence for infinite samples, given relatively arbitrary point sets. Polynomial chaos expansions sidestep some of these issues by designing sample grids \citep{Xiu2005, Nobile2007, Constantine2012, Conrad2013} with respect to the prior distribution, which are known to induce a convergent approximation of the posterior density \citep{Marzouk2009a}. However, only using prior information is likely to be inefficient; whenever the data are informative, the posterior concentrates on a small fraction of the parameter space relative to the prior \citep{LiMarzouk2014}. Figure \ref{fig:expDesignCartoon} illustrates the contrast between a prior-based sparse grid \citep{Conrad2013} and a posterior-adapted, unstructured, sample set. Overall, there is a need for efficient approaches with provable convergence properties---such that one can achieve exact sampling while making full use of the surrogate model.

\begin{figure}
\centering
\subfloat[Prior-based sparse grid samples.]{
\def\svgwidth{0.4\textwidth}
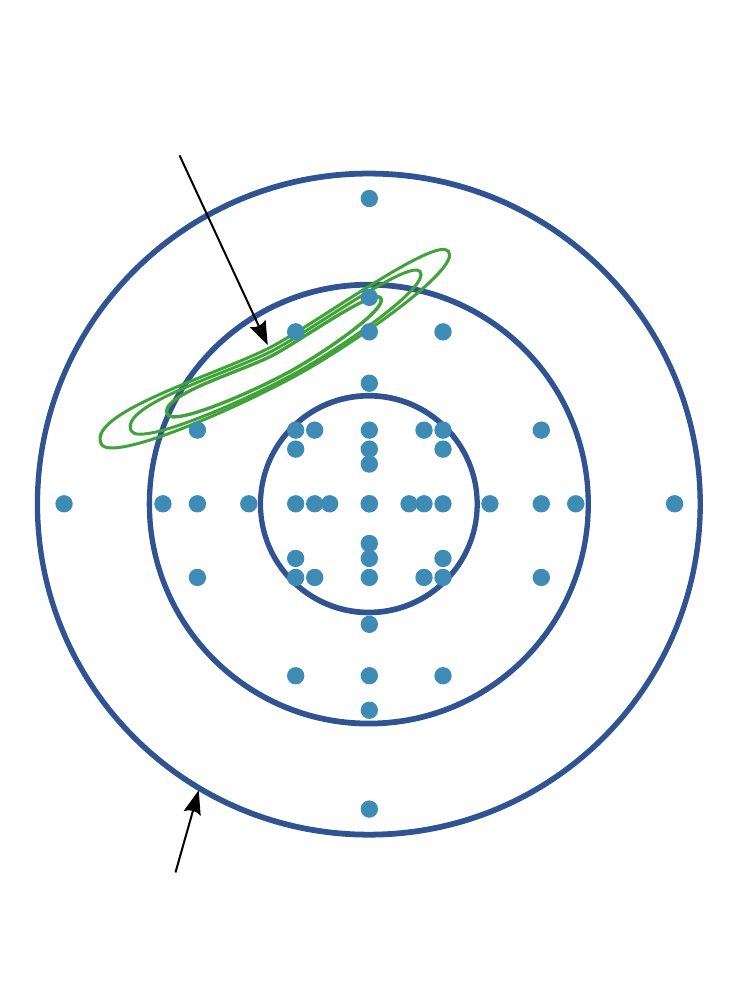
}
\subfloat[Posterior-adapted samples.]{
\def\svgwidth{0.4\textwidth}
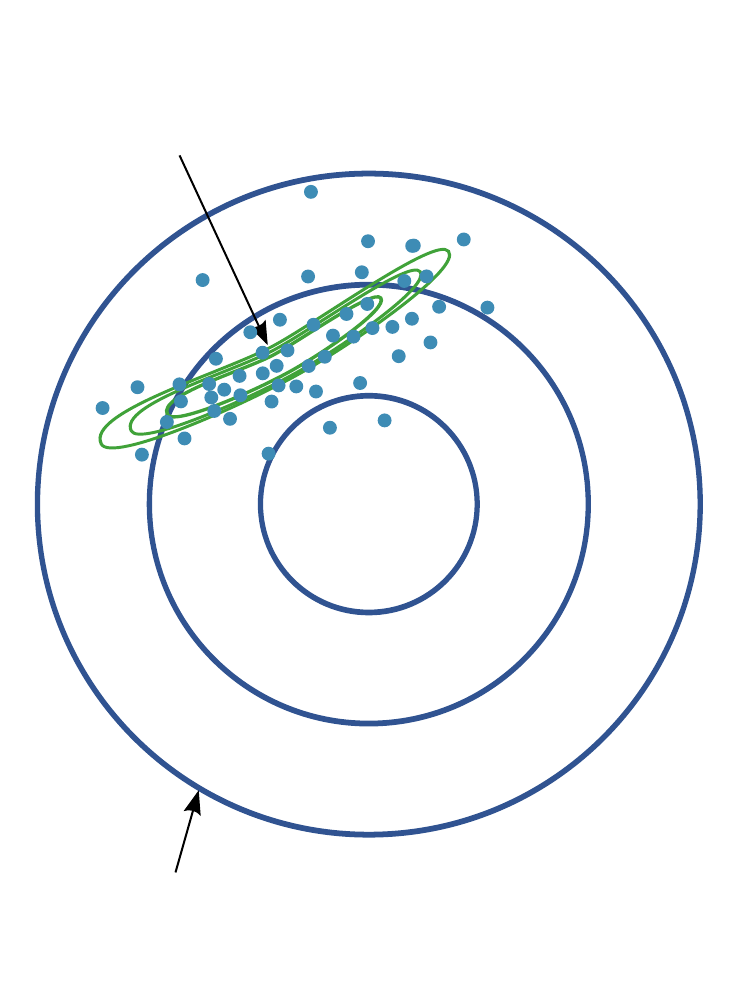

}\caption{Schematic of an inference problem with a Gaussian prior and a posterior concentrated therein, with two experimental design approaches superimposed. Points are locations in the parameter space where the forward model is evaluated.}
\label{fig:expDesignCartoon}
\end{figure}

\subsection{Our contribution}
This work attempts to resolve the above-mentioned issues by proposing a new framework that integrates \textit{local} approximations into Metropolis-Hastings kernels, producing a Markov chain that asymptotically (in the number of MCMC steps) samples from the \emph{exact} posterior distribution. As examples of this approach, we will employ approximations of either the log-likelihood function or the forward model, using local linear, quadratic, or Gaussian process regression. To produce the sample sets used for these local approximations, we will introduce a sequential experimental design procedure that interleaves infinite refinement of the approximation with the Markov chain's exploration of the posterior. The overall experimental design reflects a combination of guidance from MCMC (so that samples are focused on the posterior) and local space filling heuristics (to ensure good quality sample sets for local approximation), triggered both by random refinement and by local error indicators of approximation quality. The result is a practical approach that also permits rigorous error analysis. This concept is inspired by the use of local approximations in \emph{trust region methods} for derivative-free optimization \citep{Conn2000, Conn2009}, wherein local models similarly allow the reuse of model evaluations while enabling refinement until convergence. Local approximations also have a long history in the statistics literature \citep{Cleveland1979, Friedman1991} and have recently been reintroduced as an important strategy for scaling Gaussian processes to large data contexts \citep{Gramacy2013}.

Local approximations are convergent under relatively straightforward conditions (compared to global approximations), and we use this property to prove that the resulting MCMC algorithm converges asymptotically to the posterior distribution induced by the exact forward model and likelihood.
%
%
%
Our proof involves demonstrating that the transition kernel converges quickly as the posterior distribution is explored and as the surrogate is refined; our theoretical analysis focuses on the specific case of a random-walk Metropolis algorithm coupled with local quadratic approximations of the log-posterior density.
Our arguments are not limited to the random-walk Metropolis algorithm, however; they apply quite broadly and can be adapted to many other Metropolis-Hastings algorithms and local approximation schemes. Broadly, our theoretical results reinforce the notion that \emph{it is possible to greatly reduce the number of evaluations of the forward model per MCMC step  when the likelihood has some local regularity.} We complement the theory by demonstrating experimental performance improvements of up to several orders of magnitude on inference problems involving ordinary differential equation and partial differential equation forward models, with no discernable loss in accuracy, using several different MCMC algorithms and local approximation schemes.

%
We note that our theoretical results are asymptotic in nature; in this paper, we do not focus on finite-time error bounds. While we can comment on such bounds in a few specific settings, obtaining more general quantitative estimates for the finite-time bias of the algorithm is a significant challenge and will be tackled elsewhere. Nevertheless, we argue that asymptotic convergence is quite useful for practitioners, as it supports how the algorithm is actually applied.
Since the aim of our approach is to reduce the use of the forward model, it is natural to ask how many model runs would be necessary to construct an MCMC chain that yields estimates with a certain error. We cannot \textit{a priori} answer this question, just as we cannot (in general) say in advance how long it will take any other MCMC algorithm to reach stationarity. Yet asymptotic convergence makes our algorithm comparable to standard MCMC algorithms in practice: iterations continue until MCMC diagnostics suggest that the chain, and hence the underlying approximation, is sufficiently converged for the application. The cost of running the forward model is accumulated incrementally as the MCMC chain is extended, in a way that balances the error of the finite chain with the error introduced by the approximation. Moreover, this process may be interrupted at any time. This approach to posterior sampling stands in contrast with existing non-convergent methods, where the cost of constructing the approximation is incurred \textit{before} performing inference, and where the user must carefully balance the error induced by the approximation with the MCMC sampling error, without any rigorous strategy for doing so.

The remainder of this paper is organized as follows. We describe the new MCMC approach in Section \ref{sec:algorithm}. Theoretical results on asymptotically exact sampling are provided in Section \ref{sec:theory}; proofs of these theorems are deferred to Appendix~\ref{sec:theoryappendix}. Section \ref{sec:results} then provides empirical assessments of performance in several examples. We emphasize that, while the examples demonstrate strong computational performance, the present implementation is merely a representative of a class of asymptotically exact MCMC algorithms. Therefore, Section \ref{sec:discussion} discusses several variations on the core algorithm that may be pursued in future work. A reusable implementation of the algorithm described is available as part of the MIT Uncertainty Quantification Library, \url{https://bitbucket.org/mituq/muq/}.

\section{Metropolis-Hastings with local approximations}
\label{sec:algorithm}

This section describes our framework for Metropolis-Hastings algorithms based on local approximations, which incrementally and infinitely refine an approximation of the forward model or likelihood as inference is performed. 

\subsection{Algorithm overview}
Consider a Bayesian inference problem with posterior density
\begin{equation*}
p(\theta|\mathbf{d} ) \propto \mathcal{L}(\theta | \mathbf{d}, \mathbf{f})p(\theta),
\end{equation*}
for inference parameters $\theta \in \Theta \subseteq \mathbb{R}^d$, data $\mathbf{d} \in \mathbb{R}^n$, forward model $\mathbf{f}: \Theta \rightarrow \mathbb{R}^n$, and probability densities specifying the prior $p(\theta)$ and likelihood function $\mathcal{L}$. The forward model may enter the likelihood function in various ways. For instance, if $\mathbf{d} = \mathbf{f}(\theta) + \eta$, where $\eta \sim p_{\eta}$ represents some measurement or model error, then $\mathcal{L}(\theta |  \mathbf{d}, \mathbf{f}) = p_{\eta} ( \mathbf{d} - \mathbf{f}(\theta) )$.

A standard approach is to explore this posterior with a Metropolis-Hastings algorithm using a suitable proposal kernel $L$, yielding the Metropolis-Hastings transition kernel $K_\infty(X_t,\cdot)$; existing MCMC theory governs the correctness and performance of this approach \citep{Roberts2004}.
For simplicity, assume that the kernel $L$ is translation-invariant and symmetric.\footnote{Assuming symmetry simplifies our discussion, but the generalization to non-symmetric proposals is straightforward. Extensions to translation-dependent kernels, \emph{e.g.,} the Metropolis-adjusted Langevin algorithm, are also possible \citep{Conrad2014a}.}
We assume that the forward model evaluation is computationally expensive---requiring, for example, a high-resolution numerical solution of a partial differential equation (PDE). Also assume that drawing a proposal is inexpensive, and that given the proposed parameters and the forward model evaluation, the prior density and likelihood are similarly inexpensive to evaluate, \emph{e.g.}, Gaussian. In such a setting, the computational cost of MCMC is dominated by the cost of forward model evaluations required by $K_\infty(X_t,\cdot)$.\footnote{Identifying the appropriate target for approximation is critical to the performance of our approach, and depends upon the relative dimensionality, regularity, and computational cost of the various components of the posterior model. In most settings, the forward model is a clear choice because it contributes most of the computational cost, while the prior and likelihood may be computed cheaply without further approximation. The algorithm presented here may be adjusted to accommodate other choices by merely relabeling the terms. For another discussion of this issue, see \cite{Bliznyuk2008}.} 

Previous work has explored strategies for replacing the forward model with some cheaper approximation, and a typical scheme works as follows \citep{Rasmussen2003, Bliznyuk2012, Marzouk2007}. Assume that one has a collection of model evaluations,  $\mathcal{S} := \{(\theta, \mathbf{f}(\theta))\}$, and a method for constructing an approximation $\tilde{\mathbf{f}}$ of $\mathbf{f}$ based on those examples. This approximation can be substituted into the computation of the Metropolis-Hastings acceptance probability. However, $\mathcal{S}$ is difficult to design in advance, so the algorithm is allowed to refine the approximation, as needed, by computing new forward model evaluations near the sample path and adding them to the growing sample set $\mathcal{S}_t$. 

%
Our approach, outlined in Algorithm \ref{alg:algSketch}, is in the same spirit as these previous efforts. Indeed, the sketch in Algorithm \ref{alg:algSketch} is sufficiently general to encompass both the previous efforts mentioned above and the present work. We write $K_{t}$ to describe the evolution of the sampling process at time $t$ in order to suggest the connection of our process with a time-inhomogeneous Markov chain; this connection is made explicit in Section \ref{sec:theory}. Intuitively, one can argue that this algorithm will produce accurate samples if $\tilde{\mathbf{f}}$ is close to $\mathbf{f}$, and that the algorithm will be efficient if the size of $\mathcal{S}_t$ is small and $\tilde{\mathbf{f}}$ is cheap to construct.

\algrenewcomment[1]{\hfill\makebox[0.33\linewidth][l]{\(\triangleright\) #1}}

\begin{algorithm}
\caption{Sketch of approximate Metropolis-Hastings algorithm}
\label{alg:algSketch}
\begin{algorithmic}[1]
\Procedure{\textsc{RunChain}}{$\theta_{1}, \mathcal{S}_{1}, \mathcal{L}, \mathbf{d}, p, \mathbf{f}, L, T$}
\For{$t = 1 \ldots T$}
	\State  $(\theta_{t+1}, \mathcal{S}_{t+1}) \gets K_t(\theta_t, \mathcal{S}_t, \mathcal{L}, \mathbf{d}, p, \mathbf{f}, L)$
\EndFor
\EndProcedure
\Statex
\Procedure{$K_t$}{$\theta^-, \mathcal{S},  \mathcal{L}, \mathbf{d}, p, \mathbf{f}, L$}
	\State Draw proposal $\theta^+ \sim L(\theta^-, \cdot)$ 

\State Compute approximate models $\tilde{\mathbf{f}}^+$ and $\tilde{\mathbf{f}}^-$, \label{alg:algSketch:construct} valid near $\theta^+$ and $\theta^-$
\State Compute acceptance probability \label{alg:algSketch:alpha}
 $\alpha \gets \min \left(1,\frac{\mathcal{L}(  \theta | \mathbf{d} ,\tilde{\mathbf{f}}^+)p(\theta^+)}{\mathcal{L}(\theta  | \mathbf{d}  ,\tilde
{\mathbf{f}}^-)p(\theta^-)} \right)$ 
\If{approximation needs refinement near $\theta^-$ or $\theta^+$}
\State Select new point $\theta^\ast$ and grow $\mathcal{S} \gets \mathcal{S} \cup (\theta^\ast, \mathbf{f}(\theta^\ast))$. Repeat from Line \ref{alg:algSketch:construct}.
\Else
\State Draw $u \sim \text{Uniform}(0,1)$. If $u < \alpha$, \textbf{return} $(\theta^+, \mathcal{S})$, else \textbf{return} $(\theta^-, \mathcal{S})$.
\EndIf
\EndProcedure

\end{algorithmic}
\end{algorithm}

Our implementation of this framework departs from previous work in two important ways. First, rather than using global approximations constructed from the entire sample set $\mathcal{S}_t$, we construct local approximations that use only a nearby subset of $\mathcal{S}_t$ for each evaluation of $\tilde{\mathbf{f}}$, as in LOESS \citep{Cleveland1979} or derivative-free optimization \citep{Conn2009}. Second, previous efforts usually halt the growth of $\mathcal{S}_t$ after a fixed number of refinements;\footnote{For example, \cite{Rasmussen2003} and \cite{Bliznyuk2012} only allow refinements until some fixed time $T_{\text{ref}} < T$, and polynomial chaos expansions are typically constructed in advance, omitting refinement entirely \citep{Marzouk2007}.}
instead, we allow an infinite number of refinements to occur as the MCMC chain proceeds. Figure \ref{fig:refinementCartoon} depicts how the sample set might evolve as the algorithm is run, becoming denser in regions of higher posterior probability, allowing the corresponding local approximations to use ever-smaller neighborhoods and thus to become increasingly accurate. Together, these two changes allow us to construct an MCMC chain that, under appropriate conditions, asymptotically samples from the exact posterior. Roughly, our theoretical arguments (in Section~\ref{sec:theory} and Appendix~\ref{sec:theoryappendix}) will show that refinements of the sample set $\mathcal{S}_t$ produce a convergent approximation $\tilde{\mathbf{f}}$ and hence that $K_t$ converges to the standard ``full model'' Metropolis kernel $K_\infty$ in such a way that the chain behaves as desired.
Obviously, we require that $\mathbf{f}$ be sufficiently regular for local approximations to converge. For example, when using local quadratic approximations, it is sufficient (but not necessary) for the Hessian of $\mathbf{f}$ to be Lipschitz continuous \citep{Conn2009}.

\begin{figure}
\centering
\subfloat[Early times.]{
\def\svgwidth{0.4\textwidth}
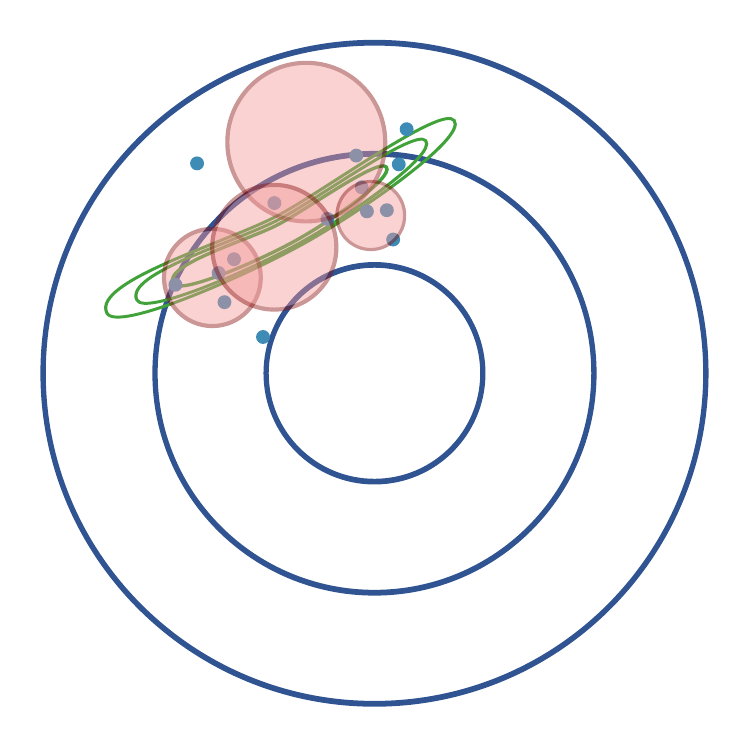}
\subfloat[Late times.]{
\def\svgwidth{0.4\textwidth}
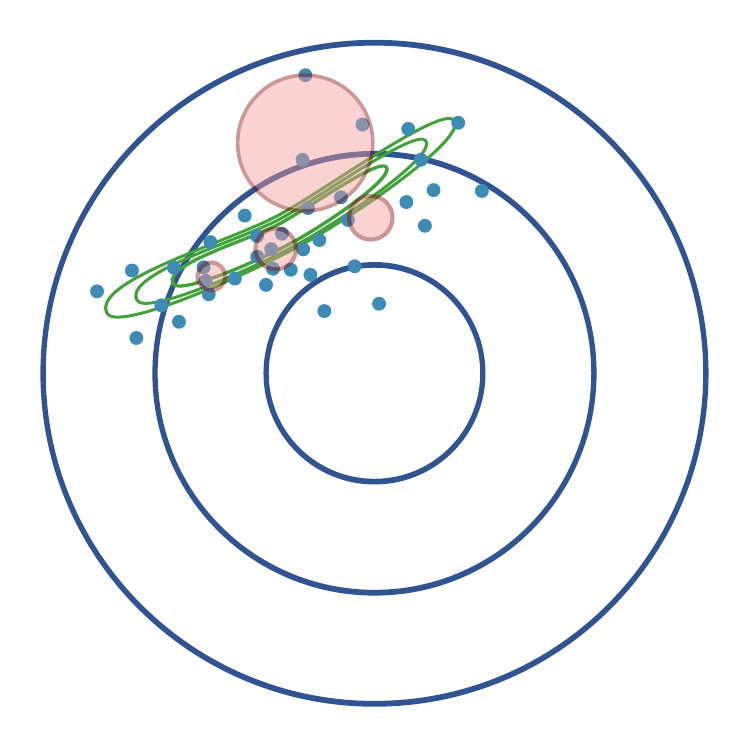}
\caption{Schematic of the behavior of local approximations as the algorithm proceeds on the example from Figure \ref{fig:expDesignCartoon}. The balls are centered at locations where local models might be needed and the radius indicates the size of the sample set; the accuracy of local models generally increases as this ball size shrinks. At early times the sample set is sparse and the local approximations are built over relatively large balls, implying that their accuracy is limited. At later times refinements enrich the sample set near regions of high posterior density, allowing the local models to shrink and become more accurate.}
\label{fig:refinementCartoon}
\end{figure}

The remainder of this section expands this outline into a usable algorithm, detailing how to construct the local approximations, when to perform refinement, and how to select new points to refine the approximations. Section \ref{s:polyapprox} describes how to construct local linear or quadratic models and outlines the convergence properties that make them useful.  Section \ref{sec:trigger} explains when to trigger refinement, either randomly or based on a cross validation error indicator. Section \ref{s:howtorefine} explains how to refine the approximations by evaluating the full model at a new point chosen using a space filling experimental design. Finally, Section \ref{sec:localgp} explains the changes required to substitute local Gaussian process approximations for polynomial approximations.

\subsection{Local polynomial approximation}
\label{s:polyapprox}
This section describes how to construct local  linear or quadratic models. We construct these models using samples from $\mathcal{S}$ drawn from a ball of radius $R$ centered on $\theta$, $\mathcal{B}(\theta,R):=  \left \{ \left (\theta_i,\mathbf{f}(\theta_i) \right ) \in \mathcal{S} : \| \theta_i - \theta\|_2 \leq R  \right \}$. If this set contains a sufficient number of samples, local polynomial models may easily be fit using least squares regression. We write the operators that produce such linear or quadratic approximations as $\mathcal{L}_{\mathcal{B}(\theta,R)}^{\sim j}$ or $\mathcal{Q}_{\mathcal{B}(\theta,R)}^{\sim j}$, respectively. The superscript $\sim \! j$, if non-empty, indicates that sample $j$ should be omitted; this option is used to support cross-validation error indicators, described below. 

It can be shown that the following error bounds hold independently for linear or quadratic approximations of each output component $i=1\ldots n$, for every point within the ball, $\theta^\prime :  \|\theta^\prime - \theta\|_2 \leq R$ \citep{Conn2009}, assuming that the gradient or Hessian of $\mathbf{f}$ is Lipschitz continuous, respectively: 
\begin{subequations}
\label{e:bounds}
\begin{eqnarray}
\left| {f_i}(\theta^\prime) - \left( \mathcal{L}^{\sim j}_{\mathcal{B}(\theta,R)}(\theta^\prime)\right)_i \right| &\leq  & \kappa_l(\nu_1, \lambda, d) R^2,\\
\left| {f_i}(\theta^\prime) - \left( \mathcal{Q}^{\sim j}_{\mathcal{B}(\theta,R)}(\theta^\prime)\right)_i \right| &\leq & \kappa_q(\nu_2, \lambda, d) R^3.
\end{eqnarray}
\end{subequations}
%
%
where the constants $\kappa$ are functions of the Lipschitz constants $\nu_1, \nu_2 < \infty$ of the gradient or Hessian of $\mathbf{f}$, respectively; a ``poisedness'' constant $\lambda$ reflecting the geometry of the input sample set; and the parameter dimension $d$. Intuitively, $\lambda$ is small if the points are well separated, fill the ball from which they are drawn, and do not lie near any linear or quadratic paths (for the linear and quadratic approximations, respectively). As long as $\lambda$ is held below some fixed finite value, the model is said to be $\lambda$-\emph{poised}, and these bounds show that the approximations converge as $R \to 0$.\footnote{Although \cite{Conn2009} explicitly compute and control the value of $\lambda$, this step is not necessary in practice for our algorithm. The geometric quality of our sample sets is generally good because of the experimental design procedure we use to construct them. Also, we are less sensitive to poor geometry because we perform regression, rather than interpolation, and because the cross validation procedure described below considers geometric quality and can trigger refinement as needed.} These simple but rigorous local error bounds form the foundation of our theoretical analysis, and are the reason that we begin with local polynomial approximations. Usefully, they are representative of the general case, in that most reasonable local models converge in some sense as the ball size falls to zero.

It remains to precisely specify the choice of radius, $R$, and the weights used in the least squares regression. The radius $R$ is selected to include a fixed number of points $N$. A linear model is fully defined by $N_{\text{def}}=d+1$ points and a quadratic is defined by $N_{\text{def}}=(d+1)(d+2)/2$ points; hence, performing a least squares regression requires at least this many samples. Such models are interpolating, but the associated least squares system is often poorly conditioned unless the geometry of the sample set is carefully designed. \cite{Conn2009} show that adding additional samples can only stabilize the regression problem, so we select $N = \sqrt{d}N_{\text{def}}$, which seems to work well in practice.\footnote{In very low dimensions, $\sqrt{d}$ provides very few extra samples and hence should be inflated. For $d=6$, in the numerical experiments below, this exact form is used.}

We depart from Conn \etal \citeyearpar{Conn2009} by performing a weighted regression using a variation of the tricube weight function often used with LOESS \citep{Cleveland1979}. If the radius that contains the inner $N_{\text{def}}$ samples is $R_{\text{def}}$, then $R > R_{\text{def}}$ and the weight of each sample is:
\begin{equation}
w_i = \begin{cases}
1 & \|\theta_i - \theta\|_2 \leq R_{\text{def}},\\
0 & \|\theta_i - \theta\|_2 > R,\\
\left(1- \left( \frac{ \|\theta_i - \theta\|_2 - R_{\text{def}}}{R-R_{\text{def}}} \right)^3 \right) ^3 & \text{else.}
\end{cases}
\label{eqn:tricubeWeights}
\end{equation}
Setting the inner points to have unity weight ensures that the regression is full rank, while subsequently decreasing the weights to zero puts less emphasis on more distant samples. An interesting side effect of using this weight function is that the global approximation $\tilde{\mathbf{f}}$ has two continuous derivatives, even though it is constructed independently at each point \citep{Atkeson1997}.

This process is described by the subroutine \textsc{LocApprox} in Algorithm \ref{alg:localApproximation}, which produces an approximation at $\theta$, using a fixed set of samples $\mathcal{S}$, optionally omitting sample $j$.  The pseudocode uses $\mathcal{A}_{\mathcal{B}(\theta,R)}^{\sim j}$ to represent either polynomial fitting algorithm. Appendix \ref{sec:polyDetails} describes the regression procedure and the numerical approach to the corresponding least squares problems in more detail. Multiple outputs are handled by constructing a separate approximation for each one. Fortunately, the expensive step of the least squares problem is identical for all the outputs, so the cost of constructing the approximation scales well with the number of observations.

\algrenewcomment[1]{\hfill\makebox[0.4\linewidth][l]{\(\triangleright\) #1}}

\begin{algorithm}
\caption{Construct local approximation}
\label{alg:localApproximation}
\begin{algorithmic}[1]

\Procedure{\textsc{LocApprox}}{$\theta, \mathcal{S}, j$}
\State Select $R$ so that $|\mathcal{B}(\theta,R)| = N$, where 
\Statex $\mathcal{B}(\theta,R) :=  \{(\theta_i,\mathbf{f}(\theta_i)) \in \mathcal{S} : \| \theta_i - \theta\|_2 \leq R\}$ \Comment{Select ball of points}
\State $\tilde{\mathbf{f}} \gets \mathcal{A}_{\mathcal{B}(\theta,R)}^{\sim j} $ \label{alg:algOverview:regress}\Comment{Local approximation as defined in }
\Statex \hfill\makebox[0.4\linewidth][l]{Section \ref{s:polyapprox}, possibly without sample $j$}
\State \textbf{return} $\tilde{\mathbf{f}}$
\EndProcedure

\end{algorithmic}
\end{algorithm}

\subsection{Triggering model refinement}
\label{sec:trigger}
We separate the model refinement portion of the algorithm into two stages. This section discusses \textit{when} refinement is needed, while Section~\ref{s:howtorefine} explains \textit{how} the refinement is performed. The MCMC step uses local approximations at both $\theta^+$ and $\theta^-$, and either are candidates for refinement. We choose a refinement criteria that is symmetric, that is, which behaves identically if the labels of $\theta^+$ and $\theta^-$ are reversed; by treating the two points equally, we aim to avoid adverse coupling with the decision of whether to accept a move. 

Refinement is triggered by either of two criteria. The first is random: with probability $\beta_t$, the model refined at either the current point $\theta^-$ or the proposed point $\theta^+$. This process fits naturally into MCMC and is essential to establishing the theoretical convergence results in the next section. The second criterion, based on a cross-validation error indicator, is intended to make the approximation algorithm efficient in practice. For a Metropolis-Hastings algorithm with a symmetric proposal, recall that the acceptance probability computed using the true forward model is
\begin{equation*}
\alpha = \min \left(1, \frac{\mathcal{L}(  \theta^+ | \mathbf{d},{\mathbf{f}})p(\theta^+)}{\mathcal{L}( \theta^- |\mathbf{d},{\mathbf{f}})p(\theta^-)} \right).
\end{equation*}
Since the acceptance probability is a scalar, and this equation is the only appearance of the forward model in the sampling algorithm, it is a natural target for an error indicator. We employ a leave-one-out cross validation strategy, computing the sensitivity of the acceptance probability to the omission of samples from each of the approximate models, producing scalar error indicators $\epsilon^+$ and $\epsilon^-$. Refinement is performed whenever one of these indicators exceed a threshold $\gamma_t$, at the point whose error indicator is larger. 

To construct the indicators, begin by computing the ratio inside the acceptance probability, using the full sample sets and variations leaving out each sample,  $j=1,\ldots, N$.
\begin{eqnarray*}
\zeta &:=& \frac{\mathcal{L}(\theta^+|\mathbf{d}, \textsc{LocApprox}(\theta^+, \mathcal{S}, \emptyset))p(\theta^+)}
{\mathcal{L}(\theta^-|\mathbf{d}, \textsc{LocApprox}(\theta^-, \mathcal{S}, \emptyset))p(\theta^-)} \\
\zeta^{+,\sim j} &:=& \frac{\mathcal{L}(\theta^+|\mathbf{d}, \textsc{LocApprox}(\theta^+, \mathcal{S}, j))p(\theta^+)}
{\mathcal{L}(\theta^-|\mathbf{d}, \textsc{LocApprox}(\theta^-, \mathcal{S}, \emptyset))p(\theta^-)}\\
\zeta^{-,\sim j} &:=& \frac{\mathcal{L}(\theta^+|\mathbf{d}, \textsc{LocApprox}(\theta^+, \mathcal{S}, \emptyset))p(\theta^+)}
{\mathcal{L}(\theta^-|\mathbf{d}, \textsc{LocApprox}(\theta^-, \mathcal{S}, j))p(\theta^-)}
\end{eqnarray*}
Next, find the maximum difference between the $\alpha$ computed using $\zeta$ and that computed using the leave-one-out variations $\zeta^{+,\sim j}$ and $\zeta^{-,\sim j}$. The error indicators consider the acceptance probability in both the forward and reverse directions, ensuring equivalent behavior under relabeling of $\theta^+$ and $\theta^-$; this prevents the cross validation process from having any impact on the reversibility of the transition kernel.
\begin{align}
\epsilon^+ &:=& \underset{j}{\max} \left( \bigg|\min \left(1, \zeta \right) - \min \left(1,\zeta^{+,\sim j}\right)\bigg| + \left|\min \left(1, \frac{1}{\zeta}\right) - \min \left(1,\frac{1}{\zeta^{+,\sim j}}\right)\right| \right) \label{eq:cvStart}\\
\epsilon^- &:=& \underset{j}{\max} \left( \bigg|\min \left(1, \zeta \right) - \min \left(1,\zeta^{-,\sim j}\right)\bigg| + \left|\min \left(1, \frac{1}{\zeta}\right) - \min \left(1,\frac{1}{\zeta^{-,\sim j}}\right)\right| \right) \label{eq:cvEnd}
\end{align}

We emphasize that the acceptance probability is a natural quantity of interest in this context; it captures the entire impact of the forward model and likelihood on the MH kernel. The cross-validation error indicator is easily computable, summarizes a variety of error sources, and is easily interpretable as an additive error in a probability. These features make it possible for the user to exercise a problem-independent understanding of the threshold to which it is compared, $\gamma_t$. In contrast, attempting to control the error in either the forward model outputs or log-likelihood at the current or proposed point is not generically feasible, as their scale and the sensitivity of the MH kernel to their perturbations cannot be known \textit{a priori}. 

Our two refinement criteria have different purposes, and both are useful to ensure a quick and accurate run. The cross validation criterion is a natural and efficient way to refine our estimates, and is the primary source of refinement during most runs. The random criterion is less efficient, but some random evaluations may be required for the algorithm to be asymptotically correct for all starting positions. Thus, we use both in combination. The two parameters $\beta_t$ and $\gamma_t$ are allowed to decrease over time, decreasing the rate of random refinement and increasing the stringency of the cross validation criterion; theory governing the rates at which they may decrease and guidance on choosing them in practice are discussed later.


\subsection{Refining the local model}
\label{s:howtorefine}
If refinement of the local model at a point $\theta$ is required, we perform refinement by selecting a single new nearby point $\theta^\ast$, computing $\mathbf{f}(\theta^\ast)$, and inserting the new pair into $\mathcal{S}$. To be useful, this new model evaluation should improve the sample set for the local model $\mathcal{B}(\theta,R)$, either by allowing the radius $R$ to decrease or by improving the local geometry of the sample set. Consider that MCMC will revisit much of the parameter space many times, hence our algorithm must ensure that local refinements maintain the global quality of the sample set, that is, the local quality at every nearby location.

Intuitively, local polynomial regression becomes ill-conditioned if the points do not fill the whole ball, or if some points are clustered much more tightly than others. The obvious strategy of simply adding $\theta$ to $\mathcal{S}$ is inadvisable because it often introduces tightly clustered points, inducing poorly conditioned regression problems. Instead, a straightforward and widely used type of experimental design is to choose points in a space-filling fashion; doing so near $\theta$ naturally fulfills our criteria. Specifically, we select the new point $\theta^\ast$ by finding a local maximizer of the problem: 
\begin{eqnarray*}
  \theta^\ast & = & \argmax_{\theta^\prime}  \, \: \min_{\theta_i \in \mathcal{S}} \| \theta^\prime - \theta_i \|_2,\\
   &	& \text{subject to } \| \theta^\prime - \theta \|_2 \leq R,
\end{eqnarray*}
where optimization iterations are initialized at $\theta^\prime=\theta$. The constraint ensures that the new sample lies in the ball and thus can be used to improve the current model, and the inner minimization operator finds a point well separated from the entire set $\mathcal{S}$ in order to ensure the sample's global quality. Inspection of the constraints reveals that the inner minimization may be simplified to $\theta_i \in \mathcal{B}(\theta, 3R)$, as points outside a ball of radius $3R$ have no impact on the optimization. We seek a local optimum of the objective because it is both far easier to find than the global optimum, and is more likely to be useful: the global optimum will often be at radius $R$, meaning that the revised model cannot be built over a smaller ball. This strategy is summarized in Algorithm \ref{alg:spaceFill}.

\algrenewcomment[1]{\hfill\makebox[0.33\linewidth][l]{\(\triangleright\) #1}}

\begin{algorithm}
\caption{Refine a local approximation}
\label{alg:spaceFill}
\begin{algorithmic}[1]
\Statex
\Procedure{\textsc{RefineNear}}{$\theta, \mathcal{S}$}
\State Select $R$ so that $|\mathcal{B}(\theta,R)| = N$ \Comment{Select ball of points}
\State $\theta^\ast \gets \argmax_{\|\theta^\prime - \theta\| \leq R} \min_{\theta_i \in \mathcal{S}} \| \theta^\prime - \theta_i\| $\Comment{Optimize near $\theta$}
\State $\mathcal{S} \gets \mathcal{S} \cup \{\theta^\ast, \mathbf{f}(\theta^\ast) \}$ \Comment{Grow the sample set}
\State \textbf{return} $\mathcal{S}$
\EndProcedure

\end{algorithmic}
\end{algorithm}

Although there is a close relationship between the set of samples where the forward model is evaluated and the posterior samples that are produced by MCMC, they are distinct and in general the two sets do not overlap. A potential limitation of the space filling approach above is that it might select points outside the support of the prior. This is problematic only if the model is not feasible outside the prior, in which case additional constraints can easily be added.

\subsection{Local Gaussian process surrogates}
\label{sec:localgp}
Gaussian process (GP) regression underlies an important and widely used class of computer model surrogates, so it is natural to consider its application in the present local approximation framework \cite{Sacks1989, Santner2003}. Local Gaussian processes have been previously explored in \citep{Vecchia1988, Cressie1991, Stein2004, Snelson2007, Gramacy2013}. This section explains how local Gaussian process approximations may be substituted for the polynomial approximations described above.

The adaptation is quite simple: we define a new approximation operator $\mathcal{G}_{\mathcal{S}}^{\sim j}$ that may be substituted for the abstract operator $\mathcal{A}_{\mathcal{B}(\theta,R)}^{\sim j}$ in Algorithm \ref{alg:localApproximation}. The error indicators are computed much as before, except that we use the predictive distribution $\tilde{\mathbf{f}}(\theta) \sim \mathcal{N}(\mu(\theta), {{\sigma}}^2(\theta))$ instead of a leave-one-out procedure. We define $\mathcal{G}_{\mathcal{S}}^{\sim j}$ to be the mean of the local Gaussian process, $\mu(\theta)$, when $j = \emptyset$, and a draw from the Gaussian predictive distribution otherwise. This definition allows us to compute $\epsilon^+$ and $\epsilon^-$ without further modification, using the posterior distribution naturally produced by GP regression.

Our implementation of GPs borrows heavily from \cite{Gramacy2013}, using a separable squared exponential covariance kernel (i.e., with a different correlation length $\ell_i$ for each input dimension) and an empirical Bayes approach to choosing the kernel hyperparameters, i.e., using optimization to find the mode of the appropriate posterior marginals. The variance is endowed with an inverse-gamma hyperprior and a MAP estimate is found analytically, while the correlation lengths and nugget are endowed with gamma hyperpriors whose product with the marginal likelihood is maximized numerically. Instead of constructing the GP only from nearest neighbors $\mathcal{B}(\theta,R)$, we use a subset of $\mathcal{S}$ that mostly lies near the point of interest but also includes a few samples further away. This combination is known to improve surrogate quality over a pure nearest-neighbor strategy \citep{Gramacy2013}. We perform a simple approximation of the strategy developed by Gramacy and Apley: beginning with a small number of the nearest points, we estimate the hyperparameters and then randomly select more neighbors to introduce into the set, where the existing samples are weighted by their distance under the norm induced by the current length scales. This process is repeated in several batches, until the desired number of samples is reached. We are relatively unconstrained in choosing the number of samples $N$; in the numerical examples to be shown later, we choose $N=d^{5/2}$, mimicking the choice for quadratic approximations. Multiple outputs are handled with separate predictive distributions, but the hyperparameters are jointly optimized.\footnote{Choosing an optimal number of samples is generally challenging, and we do not claim that this choice of $N$ is the most efficient. Rather, it is the same scaling that we use for local quadratic approximations, and appears to work well for GP approximation in the range where we have applied it. For very low $d$, however, this $N$ may need to be increased.}

\subsection{Algorithm summary}
Our Metropolis-Hastings approach using local approximations is summarized in Algorithm \ref{alg:algOverview}. The algorithm proceeds in much the same way as the sketch provided in Algorithm \ref{alg:algSketch}. It is general enough to describe both local polynomial and Gaussian process approximations, and calls several routines developed in previous sections. The chain is constructed by repeatedly constructing a new state with $K_t$.\footnote{Before MCMC begins, $\mathcal{S}_1$ needs to be seeded with a sufficient number of samples for the first run. Two simple strategies are to draw these samples from the prior, or else near the MCMC starting point, which is often the posterior mode as found by optimization.} This function first draws a proposal and forms the approximate acceptance probability. Then error indicators are computed and refinement is performed as needed, until finally the proposal is accepted or rejected.

\algrenewcomment[1]{\hfill\makebox[0.33\linewidth][l]{\(\triangleright\) #1}}

\begin{algorithm}
\caption{Metropolis-Hastings with local approximations}
\label{alg:algOverview}
\begin{algorithmic}[1]

\Procedure{\textsc{RunChain}}{$\mathbf{f}, L, \theta_{1}, \mathcal{S}_{1}, \mathcal{L}, \mathbf{d}, p, T, \{ \beta_t \}_{t=1}^{T}, \{ \gamma_t \}_{t=1}^{T}$}
\For{$t = 1 \ldots T$}
	\State  $(\theta_{t+1}, \mathcal{S}_{t+1}) \gets K_t(\theta_t, \mathcal{S}_t, \mathcal{L}, \mathbf{d}, p, \mathbf{f}, L, \beta_t, \gamma_t)$
\EndFor
\EndProcedure

\Statex

\Procedure{$K_t$}{$\theta^-, \mathcal{S},  \mathcal{L}, \mathbf{d}, p, \mathbf{f}, L,\beta_t, \gamma_t$}
	\State Draw proposal $\theta^+ \sim L(\theta^-, \cdot)$ 

	\State $\tilde{\mathbf{f}}^+ \gets \textsc{LocApprox}(\theta^+, \mathcal{S}, \emptyset)$  \label{alg:algOverview:fplus} \Comment{Compute nominal approximations}
	\State $\tilde{\mathbf{f}}^- \gets \textsc{LocApprox}(\theta^-, \mathcal{S}, \emptyset)$ \label{alg:algOverview:fminus} 
	\State $\alpha \gets \min \left(1,\frac{\mathcal{L}(  \theta | \mathbf{d} ,\tilde{\mathbf{f}}^+)p(\theta^+)}{\mathcal{L}(\theta  | \mathbf{d}  ,\tilde
{\mathbf{f}}^-)p(\theta^-)} \right)$ \Comment{Compute nominal acceptance ratio} \label{alg:algOverview:errorStart} 
\State Compute $\epsilon^+$ and $\epsilon^-$ as in Equations \ref{eq:cvStart}-\ref{eq:cvEnd}.

\If{$u \sim \text{Uniform}(0,1) < \beta_t$} \Comment{Refine with probability $\beta_t$}
	\State Randomly, $\mathcal{S} \gets \textsc{RefineNear}(\theta^+, \mathcal{S})$ or $\mathcal{S} \gets \textsc{RefineNear}(\theta^-, \mathcal{S})$
\ElsIf{$\epsilon^+ \geq \epsilon^-$ and $ \epsilon^+ \geq \gamma_t$} \Comment{If needed, refine near the larger error}  \label{alg:algOverview:refineStart}
\State $\mathcal{S} \gets \textsc{RefineNear}(\theta^+, \mathcal{S})$
\ElsIf{$\epsilon^- > \epsilon^+$ and $ \epsilon^- \geq \gamma_t$}
\State $\mathcal{S} \gets \textsc{RefineNear}(\theta^-, \mathcal{S})$

	\EndIf
	\If{refinement occured} repeat from Line \ref{alg:algOverview:fplus}.
\Else \Comment{Evolve chain using approximations}
\State Draw $u \sim \text{Uniform}(0,1)$. If $u < \alpha$, \textbf{return} $(\theta^+, \mathcal{S})$, else \textbf{return} $(\theta^-, \mathcal{S})$.
\EndIf
\EndProcedure

\end{algorithmic}
\end{algorithm}

\section{Theoretical results}
\label{sec:theory}
In this section we show that, under appropriate conditions, the following slightly modified  version of Algorithm \ref{alg:algOverview} converges to the target posterior $p(\theta|\d)$ asymptotically:
\begin{enumerate}
\item The sequence of parameters $\{ \beta_{t}\}_{t \in \mathbb{N}}$ used in that algorithm are of the form $\beta_{t} \equiv \beta > 0$. Our results hold with essentially the same proof if we use any sequence $\{ \beta_{t} \}_{t \in \mathbb{N}}$ that satisfies $\sum_{t} \beta_{t} = \infty$. Example \ref{ExDecRateBeta} in Appendix B shows that this is sharp: if $\sum_{t} \beta_{t} < \infty$, the algorithm can have a positive probability of failing to converge asymptotically, regardless of the sequence $\{ \gamma_{t} \}_{t \in \mathbb{N}}$.
\item The approximation of $\log p(\theta | \d)$ is made via quadratic interpolation on the $N = N_{\defi}$ nearest points.
We believe this to be a representative instantiation of the algorithm; similar results can be proved for other approximations of the likelihood function. 

\item The sub-algorithm \textsc{RefineNear} is replaced with:
\be 
\textsc{RefineNear}(\theta, \mathcal{S}) = \textbf{return}( \mathcal{S} \cup \{ (\theta, f(\theta)) \} ).
\ee 
This assumption substantially simplifies and shortens our argument, without substantially impacting the algorithm.
\item We fix a constant $0 < \lambda < 1$. In step 14, immediately before the word \textbf{then}, we add `\textbf{or}, for $\mathcal{B}(\theta^{+}, R)$ as defined in the subalgorithm $\textsc{LocApprox}(\theta^{+}, \mathcal{S}, \emptyset)$ used in step 8, the collection of points $\mathcal{B}(\theta^{+}, R) \cap \mathcal{S}$ is not $\lambda$-poised'. We add the same check, with $\theta^{-}$ replacing $\theta^{+}$ and `step 9' replacing `step 8', in step 16. The concept of poisedness is defined in \citep{Conn2009}, but the details are not required to read this proof. This additional check is needed for our approximate algorithm to `inherit' a one-step drift condition from the `true' algorithm. Empirically, we have found that this check rarely triggers refinement for sensible values of $\lambda$. 

\end{enumerate}

\subsection{Assumptions}

We now make some general assumptions and fix notation that will hold throughout this section and in Appendix~\ref{sec:theoryappendix}.  Denote by $\{ X_{t} \}_{t \in \mathbb{N}}$ a version of the stochastic process on $\Theta \subset \mathbb{R}^{d}$ defined by this modified version of Algorithm \ref{alg:algOverview}. Let $L(x,\cdot)$ be the kernel on $\mathbb{R}^{d}$ used to generate new proposals in Algorithm \ref{alg:algOverview}, $\ell(x,y)$ denote its density, and  $L_{t}$ be the point proposed at time $t$ in  Algorithm \ref{alg:algOverview}. Let $K_{\infty}(x,\cdot)$ be the MH kernel associated with proposal kernel $L$ and target distribution $p(\theta | \d)$. Assume that, for all measurable $A \subset \Theta$, we can write $K_{\infty}(x,A) = r(x) \delta_{x}(A) + (1 - r(x)) \int_{y \in A} p(x,y)dy$ for some $0 \leq r(x) \leq 1$ and density $p(x,y)$. Also assume that $L(x,\cdot)$ satisfies 
\be \label{eqn:Symrw}
L(x,S) = L(x+y, S+y)
\ee 
for all points $x,y \in \Theta$ and all measurable sets $S \subset \Theta$.  \par

Denote by $\mathcal{S}_{t}$ the collection of points in $\mathcal{S}$ from Algorithm \ref{alg:algOverview} at time $t$, denote by $R = R_{t}$ the value of $R_{\defi}$ at time $t$, and denote by $q_{t}^{1}, \ldots, q_{t}^{N}$ the points in $\mathcal{S}_{t}$ within distance $R_{t}$ of $X_{t}$. 

We define the \textit{Gaussian envelope condition}:
\begin{ass} \label{ass:GE}
There exists some positive definite matrix $[a_{ij}]$ and constant $0 < G < \infty$ so that the distribution
\be 
\log{p_{\infty}}(\theta_1,\theta_2, \ldots, \theta_d) = -\sum_{1 \leq i \leq j \leq d} a_{ij} \theta_{i} \theta_{j}
\ee  
satisfies
\be \label{EqGaussEnv}
\lim_{r \rightarrow \infty} \sup_{\| \theta \| \geq r} | \log p(\theta | \d) - \log{p_{\infty}(\theta)} | < G.
\ee
\end{ass}
For $\theta \in \Theta$, define the \emph{Lyapunov function}
\be \label{eqn:Vx}
V(\theta) = \frac{1}{\sqrt{p_{\infty}(\theta)}}.
\ee 
\begin{ass} \label{ass:geotarg} The proposal kernel $L$ and the density $p_{\infty}(\theta)$ satisfy the following:
\begin{enumerate} \label{assump:RoTw}
\item For all compact sets $\mathcal{A}$, there exists $\epsilon = \epsilon(\mathcal{A})$ so that $\inf_{y \in \mathcal{A}} \ell(0,y) \geq \epsilon > 0$.
\item There exist constants $C, \epsilon_0, x_{0} \geq 0$ so that $\ell(0, x) \leq C p_{\infty}(x)^{\frac{1}{1+\epsilon_0}}$ for all $\| x \| \geq x_{0}$.
\item The Metropolis-Hastings Markov chain $Z_t$ with proposal kernel $L$ and stationary density $p_\infty$ satisfies the drift condition
\be 
\E[V(Z_{t+1}) | Z_{t} = x] \leq \alpha V(x) + b
\ee 
for some $0 \leq \alpha < 1$ and some $0 \leq b < \infty$.
\end{enumerate}
\end{ass}

Before giving the main result, we briefly discuss the assumptions above.
\begin{enumerate}
\item Assumption \ref{ass:GE} is quite strong. It is chosen as a representative sufficient condition for convergence of our algorithm on unbounded state spaces primarily because it is quite easy to state and to check. The assumption is used only to guarantee that our approximation of the usual MH chain inherits a drift condition (\textit{i.e.} so that Lemma \ref{LemmaInfDrift} of Appendix B holds), and may be replaced by other assumptions that provide such a guarantee. We give some relevant alternative assumptions at the end of Appendix B. In particular, instead of Assumption \ref{ass:GE}, if we assume that the posterior $p(\theta| \d)$ has sub-Gaussian tails with bounded first and second derivatives, our methods can be reworked to show the ergodicity of a slight modification of Algorithm \ref{alg:algOverview}.  

Although Assumption \ref{ass:GE} is very strong, it does hold for one important class of distributions: mixtures of Gaussians for which one mixture component has the largest variance. That is, the condition holds if $p(\theta | \d)$ is of the form $\sum_{i=1}^{k} \alpha_{i} \mathcal{N}(\mu_{i}, \Sigma_{i})$ for some weights $\sum_{i=1}^{k} \alpha_{i} = 1$, some means $\{ \mu_{i} \}_{i=1}^{k} \in \mathbb{R}^d$, and some $d \times d$ covariance matrices $\{ \Sigma_{i} \}_{i=1}^{k}$ that satisfy $v^{\top} \Sigma_{1} v > v^{\top} \Sigma_{i} v$ for all $0 \neq v \in \mathbb{R}^{d}$ and all $i \neq 1$. 
\item Assumption \ref{ass:geotarg} holds for a very large class of commonly used Metropolis-Hastings algorithms (see, \textit{e.g.}, \cite{RoTw96} for sufficient conditions for item 3 of Assumption \ref{ass:geotarg}.) 
\end{enumerate}

\subsection{Ergodicity}

Here we state our main theorems on the convergence of the version of Algorithm \ref{alg:algOverview} introduced in this section. Proofs are given in Appendix~\ref{sec:theoryappendix}.

\begin{theorem}\label{ThmErgGaussEnv}
 Suppose Assumption \ref{ass:geotarg} holds. There exists some $G_{0} = G_{0}(L,p_{\infty}, \lambda,N)$ so that if assumption  \ref{ass:GE} holds with $0 < G < G_{0} < \infty$, then for any starting point $X_{0} = x \in \Theta$, we have
\be
\lim_{t \rightarrow \infty} \|\mathcal{L}(X_{t}) - p(\theta | \d) \|_{\TV} = 0.
\ee
\end{theorem}

If we assume that $\Theta$ is compact, the same conclusion holds under much weaker assumptions: 

\begin{theorem}  \label{ThmCompSup}
Suppose $\Theta$ is compact and that both $p(\theta | \d)$ and $\ell(x,y)$ are bounded away from 0 and infinity. Then
\be
\lim_{t \rightarrow \infty} \| \mathcal{L}(X_{t}) - p(\theta | \d) \|_{\TV} = 0.
\ee
\end{theorem}
\begin{remark}
We focus only on ergodicity, and in particular, do not obtain rates of convergence, laws of large numbers, or central limit theorems. We believe that, using results from the adaptive MCMC literature (see \cite{Fort2012}), the law of large numbers and central limit theorem can be shown to hold for the Monte Carlo estimator from our algorithm.  A significantly more challenging issue is to quantify the bias-variance tradeoff of our algorithm and its impact on computational effort. 
%
We plan to study this issue in a forthcoming paper.
\end{remark}



\section{Numerical experiments}
\label{sec:results}

Although the results in Section \ref{sec:theory} and further related results in Appendix \ref{sec:theoryappendix} establish the asymptotic exactness of our MCMC framework, it remains to demonstrate that it performs well in practice. This section describes three examples in which local surrogates produce accurate posterior samples using dramatically fewer evaluations of the forward model than standard MCMC. Additionally, these examples explore parameter tuning issues and the performance of several algorithmic variations. Though certain aspects of these examples depart from the assumptions of Theorems~\ref{ThmErgGaussEnv} or \ref{ThmCompSup}, the discussion in Appendix \ref{sec:alternateAssumptions}
 suggests that the theory is extensible to these cases; the success of the numerical experiments below reinforces this notion.

For each of these examples, we consider the accuracy of the computed chains and the number of forward model evaluations used to construct them. In the absence of analytical characterizations of the posterior, the error in each chain is estimated by comparing the posterior covariance estimates computed from a reference MCMC chain---composed of multiple long chains computed \textit{without} any approximation---to posterior covariance estimates computed from chains produced by Algorithm \ref{alg:algOverview}. The forward models in our examples are chosen to be relatively inexpensive in order to allow the construction of such chains and hence a thorough comparison with standard samplers. Focusing on the number of forward model evaluations is a problem-independent proxy for the overall running time of the algorithm that is representative of the algorithm's scaling as the model cost becomes dominant.

The first example uses an exponential-quartic distribution to investigate and select tunings of the refinement parameters $\beta_t$ and $\gamma_t$. The second and third examples investigate the performance of different types of local approximations (linear, quadratic, and Gaussian process) when inferring parameters for an ODE model of a genetic circuit and the diffusivity field in an elliptic PDE, respectively. We conclude with some brief remarks on the performance and scaling of our implementation.

\subsection{Exponential-quartic distribution}
\label{s:rosenbrock}

To investigate tunings of the the refinement parameters $\beta_t$ and $\gamma_t$, we consider a simple two dimensional target distribution, with log-density
\be
\log p(\theta) = - \frac{1}{10} \theta_1^4 - \frac{1}{2}(2\theta_2-\theta_1^2)^2,
\ee
illustrated in Figure \ref{fig:rosenbrockMarginals}. Performing MCMC directly on this model is of course very inexpensive, but we may still consider whether local quadratic approximations can reduce the number of times the model must be evaluated. For simplicity, we choose the proposal distribution to be a Gaussian random walk with variance tuned to $\sigma^2=4$. 

\begin{figure}[htbp]
\centering
\includegraphics[scale=.8]{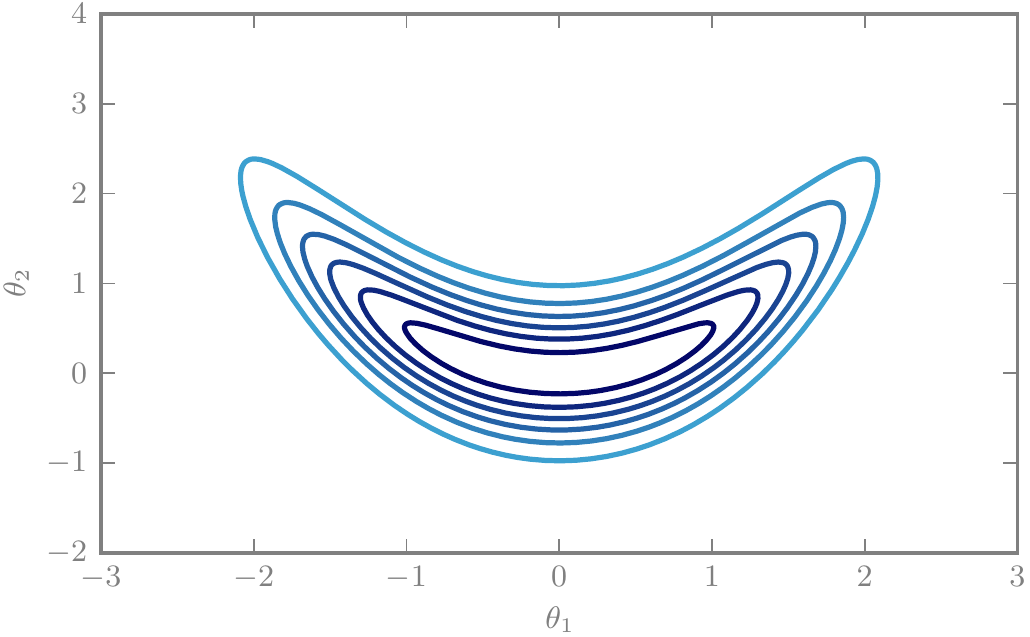}
\caption{The logarithm of the target density in the exponential-quartic example.}
\label{fig:rosenbrockMarginals}

\end{figure}

As a first step towards understanding the response of our approach to $\beta_t$ and $\gamma_t$, we test several constant values, setting only one of $\beta_n$ or $\gamma_n$ to be nonzero, choosing from $\beta_n \in \{10^{-3},10^{-2},10^{-1}\}$ and $\gamma_n \in \{10^{-2},10^{-1}, 0.5\}$. With these settings, we run Algorithm \ref{alg:algOverview}, using local quadratic approximations of the log-target density. 

The baseline configuration to which we compare Algorithm \ref{alg:algOverview} comprises 30 chains, each run for $10^5$ MCMC steps using the true forward model (\emph{i.e.}, with no approximation). In all of the numerical experiments below, we discard the first 10\% of a chain as burn-in. The reference runs are combined to produce a ``truth'' covariance, to which we compare the experiments. The chains are all initialized at the same point in the high target density region.
Ten independent chains are run for each parameter setting, with each chain containing $10^5$ MCMC steps. After discarding $10^4$ burn-in samples for each chain, we consider the evolution of the error as the chain lengthens; we compute a relative error measure at each step, consisting of the Frobenius norm of the difference in covariance estimates, divided by the Frobenius norm of the reference covariance. 
\begin{figure}[htb]
\centering
\subfloat[The accuracy of the chains.]{
\label{fig:rosenbrockConstants:a}

\includegraphics[scale=.8]{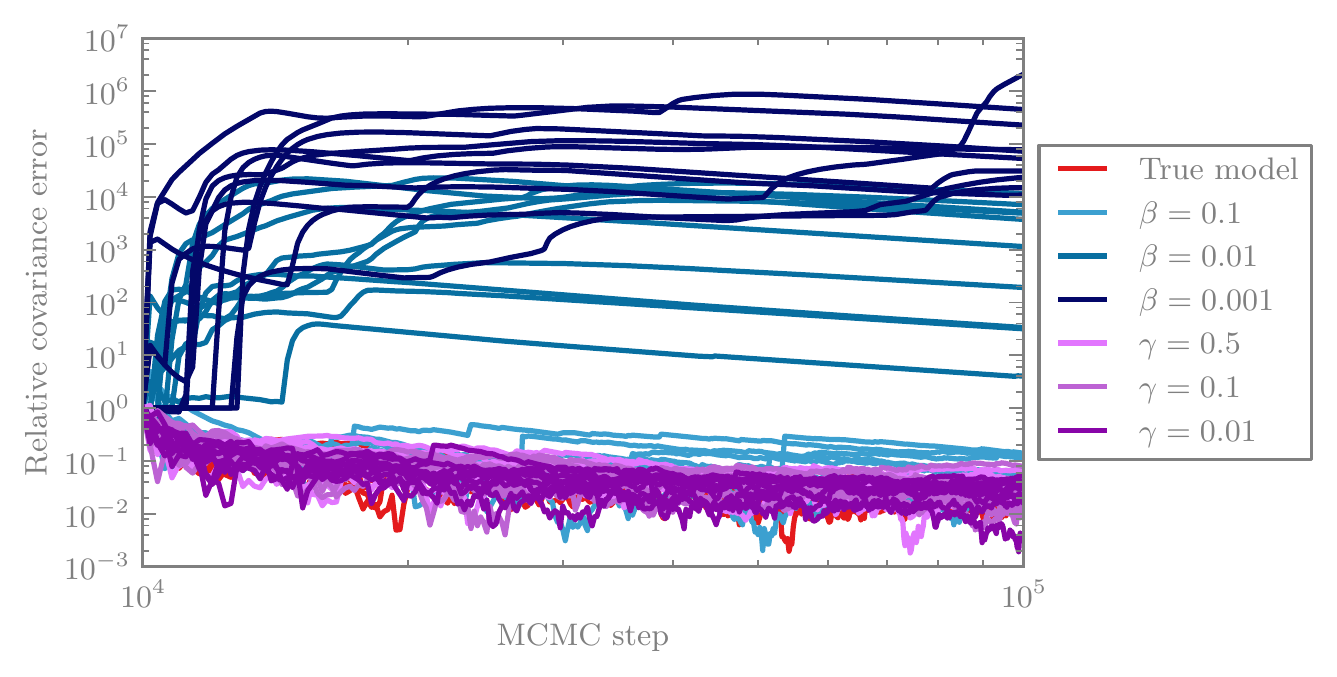}
}\\
\subfloat[The cost of the chains.]{
\label{fig:rosenbrockConstants:b}

\includegraphics[scale=.8]{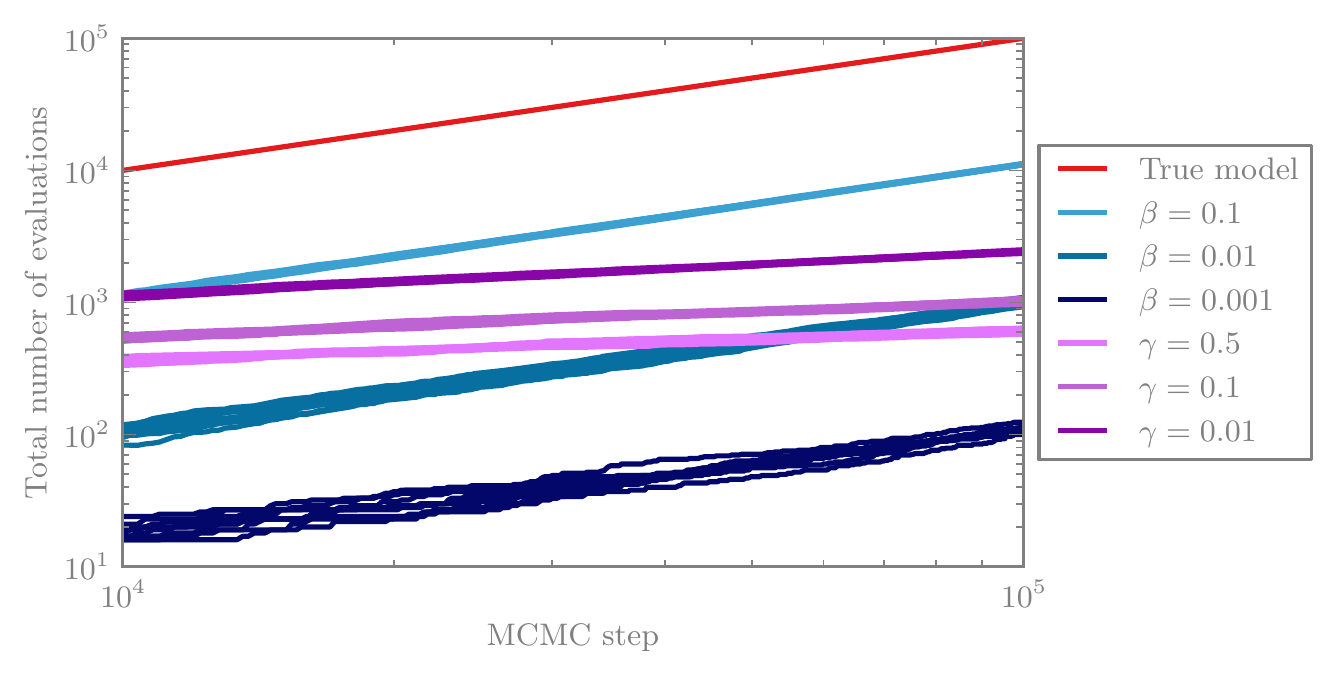}

}\caption{The accuracy and cost of sampling the exponential-quartic example using constant refinement parameters.}
\label{fig:rosenbrockConstants}
\end{figure}

This accuracy comparison is summarized in Figure~\ref{fig:rosenbrockConstants:a}, which shows the evolution of the error with the number of MCMC steps. The corresponding computational costs are summarized in Figure~\ref{fig:rosenbrockConstants:b}, which shows the number of true model evaluations performed for any given number of MCMC steps. The distribution of errors obtained with the baseline chains, shown in red, reflects both the finite accuracy of the reference chain and the variance resulting from finite baseline chain lengths. As expected, the cost of a chain increases when $\beta_t$ is larger or $\gamma_t$ is smaller; these values trigger more frequent random refinements or more strictly constrain the acceptance probability error indicator, respectively. When $\beta$-refinement is set to occur at a very low rate, the resulting chain is inexpensive but of low accuracy, and in contrast, higher values of $\beta$ show increased cost and reduced errors. The theory suggests that any constant $\beta_t>0$ should yield eventual convergence, but this difference in finite time performance is not surprising. Even the $\beta_t=0.01$ chains eventually show a steady improvement in accuracy over the interval of chain lengths considered here, which may reflect the predicted asymptotic behavior. Our experiments also show the efficacy of cross validation: all the chains using cross-validation refinement have accuracies comparable to the baseline runs while making significantly reduced use of the true model. These accuracies seem relatively insensitive to the value of $\gamma$. 

In practice, we use the two criteria jointly and set the parameters to decay with $t$. Allowing $\beta_t$ to decay is a cost-saving measure, and is theoretically sound as long as $\sum_t \beta_t$ diverges; on the other hand, setting $\gamma_t$ to decay increases the stringency of the cross validation criterion, improving robustness. Based upon our experimentation, we propose to use parameters $\beta_t = 0.01 t^{-0.2}$ and $\gamma_t = 0.1t^{-0.1}$; this seems to be a robust choice, and we use it for the remainder of the experiments. 

Figure \ref{fig:rosenbrock_final} summarizes the accuracy and cost of these parameter settings, and also considers the impact of a faster decay for the cross validation criterion: $\gamma_t=0.1t^{-0.6}$. The proposed parameters yield estimates that are comparable in accuracy to the standard algorithm, but cheaper (shifted to the left) by nearly two orders of magnitude. Observe that tightening $\gamma_t$ more quickly does not improve accuracy, but does increase the cost of the chains. 

\begin{figure}[htbp]
\centering
\includegraphics[scale=.8]{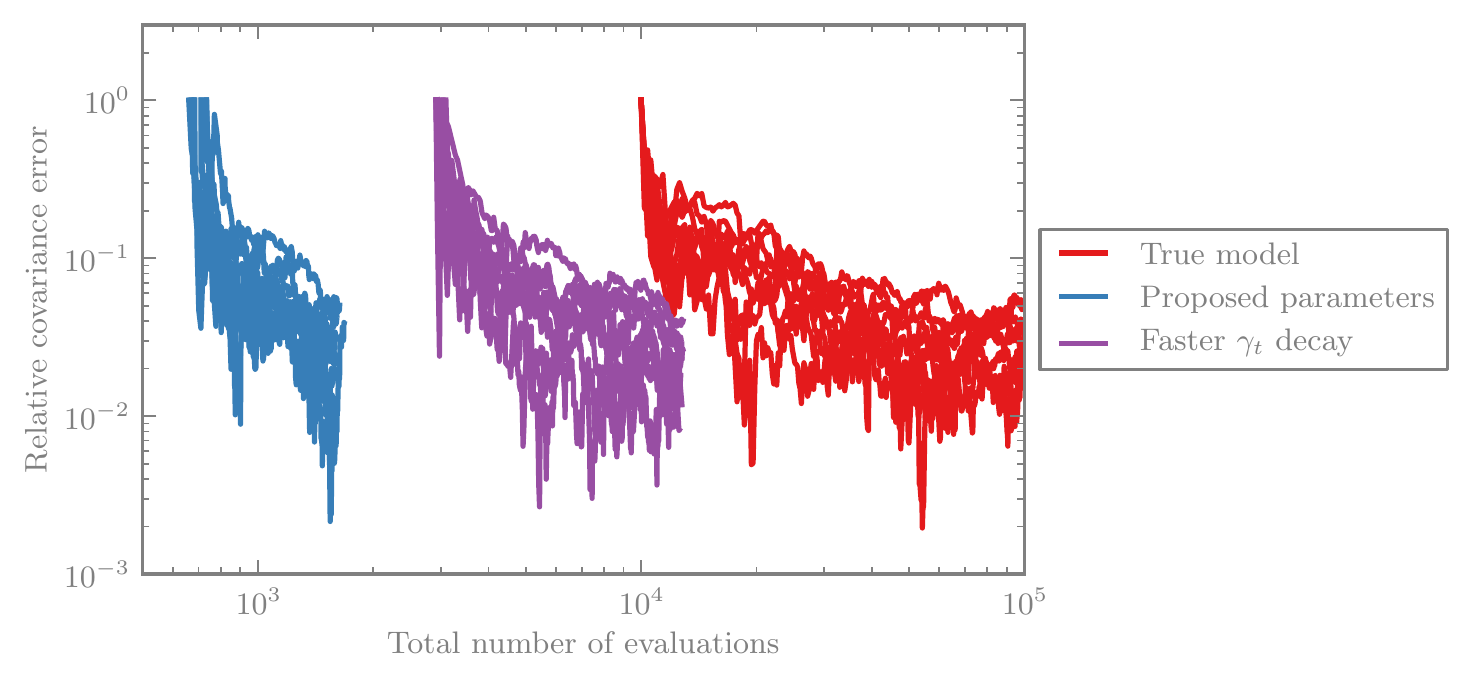}

\caption{The accuracy of the inference as compared to the number of forward model evaluations required using the proposed parameters or a setting with faster $\gamma_t$ decay. The plot depicts ten independent chains of each type, with the first $10\%$ of each chain removed as burn-in.}
\label{fig:rosenbrock_final}
\end{figure}

Before concluding this example, we explore the behavior of the refinement scheme in more detail. Figure \ref{fig:rosenbrockCostbreakdown} shows that under the proposed settings, though most refinements are triggered by cross validation, a modest percentage are triggered randomly; we propose that this is a useful balance because it primarily relies on the apparent robustness of cross validation, but supplements it with the random refinements required for theoretical guarantees. Interestingly, even though the probability of random refinement is decreasing \textit{and} the stringency of the cross-validation criterion is increasing, the proportion of refinements triggered randomly is observed to increase. This behavior suggests that the local approximations are indeed becoming more accurate as the chains progress.

\begin{figure}[htb]
\centering
\includegraphics[scale=.8]{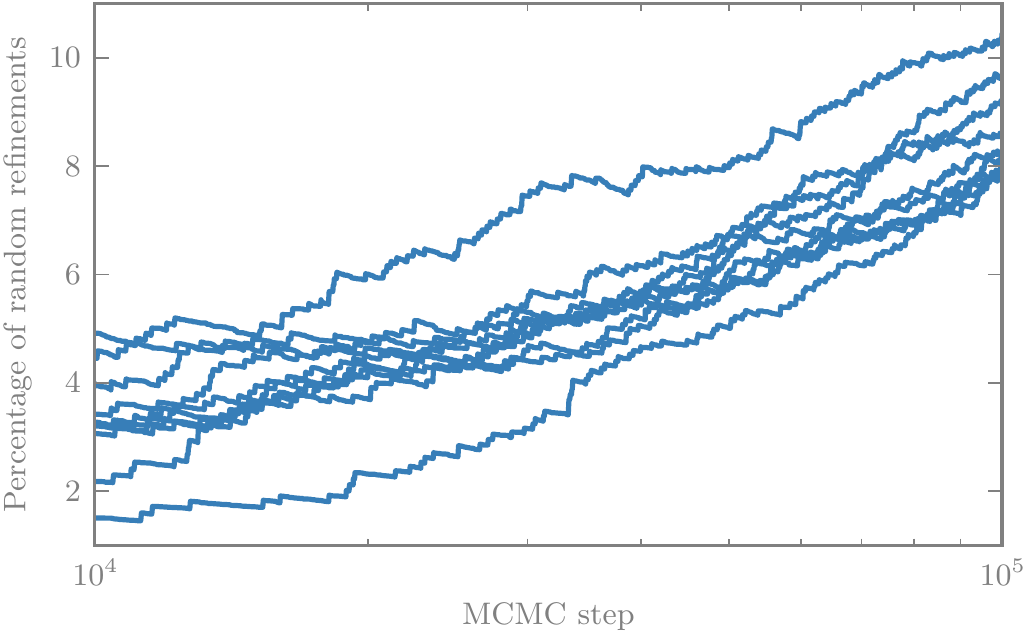}
\caption{The percentage of refinements triggered by the random refinement criterion, for ten independent chains in the exponential-quartic example, using the proposed parameters.}
\label{fig:rosenbrockCostbreakdown}

\end{figure}

Finally, it is instructive to directly plot the observed error indicators and compare them to the threshold used for refinement, as in Figure \ref{fig:rosenbrockCV}. Refinement occurs whenever the error indicators $\epsilon$, denoted by circles, exceed the current $\gamma_t$. Comparing Figures~\ref{fig:rosenbrockCVa} and \ref{fig:rosenbrockCVb}, we observe that many points lie just below the differing refinement thresholds, suggesting that choosing $\gamma_t$ provides significant control over the behavior of the algorithm.

\begin{figure}[htb]
\centering
\subfloat[The proposed parameters.]{
\includegraphics[scale=.8]{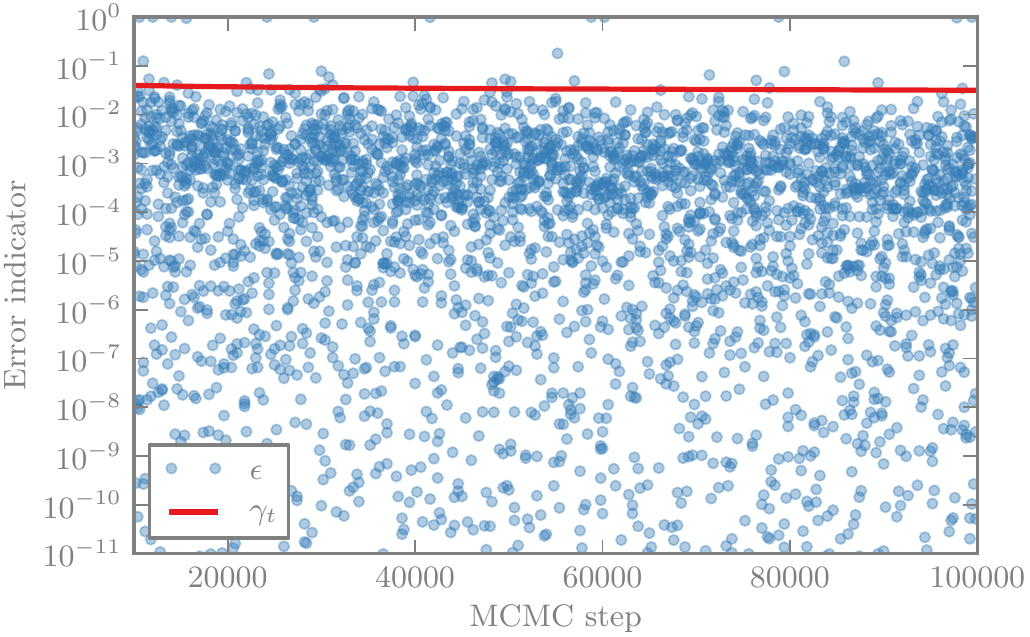}
\label{fig:rosenbrockCVa}
}\\
\subfloat[Faster decay of $\gamma_t$.]{
\includegraphics[scale=.8]{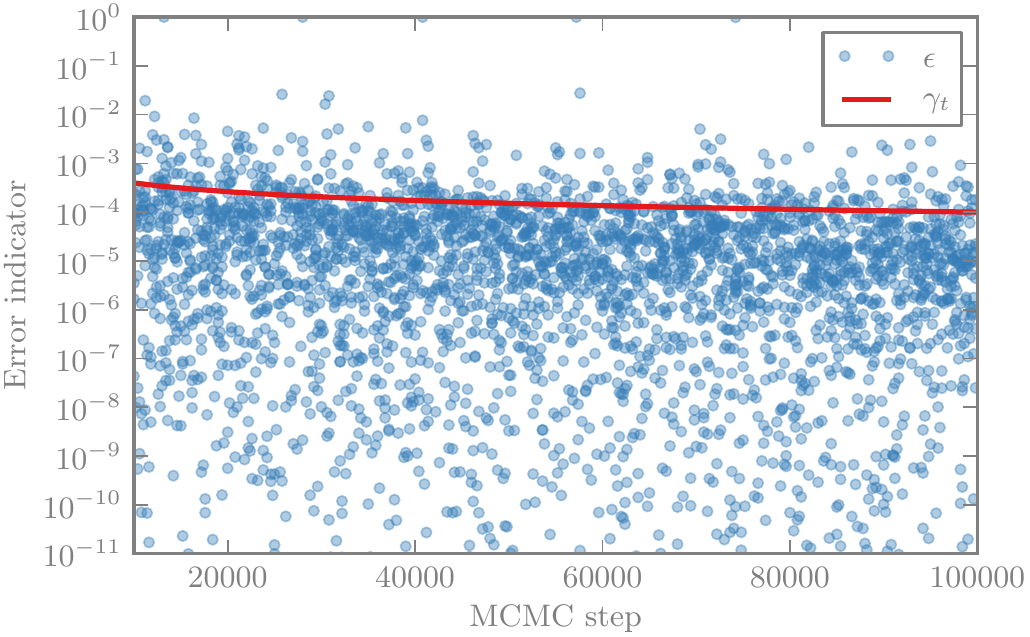}
\label{fig:rosenbrockCVb}
}\caption{The cross validation error indicator for the exponential-quartic example, using the proposed parameters or a faster $\gamma_t$ decay. The indicator shown is $\epsilon = \max(\epsilon^+, \epsilon^-)$, computed before any refinement occurs. The error indicators are often much smaller than $10^{-11}$---i.e., some proposals should obviously be accepted or rejected---but the plots are truncated to focus on behavior near the $\gamma_t$ threshold.}
\label{fig:rosenbrockCV}
\end{figure}

\subsection{Genetic toggle switch}
\label{s:genetics}
Given the refinement parameters chosen in the previous example, we now consider the performance of several different types of local approximations in an ODE model with a compact parameter domain.
We wish to infer the parameters of a genetic ``toggle switch'' synthesized in \textit{E.\ coli} plasmids by \cite{Gardner2000}, and previously used in an inference problem by \cite{Marzouk2009a}. \cite{Gardner2000} proposed a  differential-algebraic model for the switch, with six unknown parameters $Z_\theta = \{\alpha_1, \alpha_2, \beta, \gamma, K, \eta\} \in \mathbb{R}^6$, while the data correspond to observations of the steady-state concentrations. As in \cite{Marzouk2009a}, the parameters are centered and scaled around their nominal values so that they can be endowed with uniform priors over the hypercube $[-1,1]^6$. The measurement errors are independent and Gaussian, with zero mean and variances that depend on the experimental conditions. Further details on the problem setup are given in Appendix~\ref{apx:genetics}. Figure \ref{fig:geneticsMarginals} shows marginal posterior densities of the normalized parameters $\theta$. These results broadly agree with \cite{Marzouk2009a} and indicate that some directions are highly informed by the data while others are largely defined by the prior, with strong correlations among certain parameters.

\begin{figure}[htbp]
\centering
\includegraphics[scale=.8]{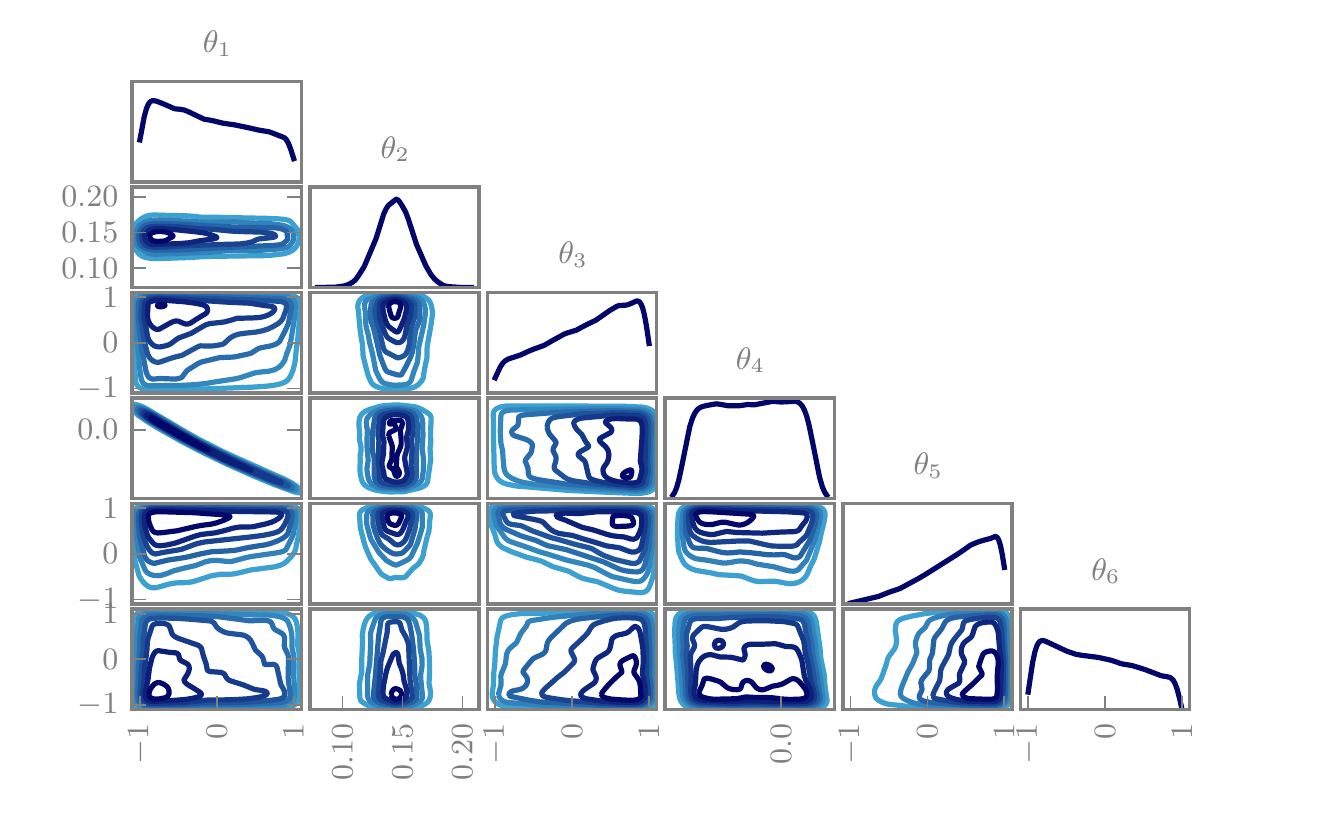}
\caption{One- and two-dimensional posterior marginals of the six parameters in the genetic toggle switch.}
\label{fig:geneticsMarginals}

\end{figure}

We investigate the performance of three different local approximations of the forward model: linear, quadratic, and Gaussian process. The experiment proceeds as in the last section (Section~\ref{s:rosenbrock}), with two differences: first, we adapt the covariance of the Gaussian proposal using the adaptive Metropolis algorithm of \cite{Haario2001}, a more practical choice than a fixed-size Gaussian random walk. Second, we limit our algorithm to perform at most two refinements per MCMC step, which is an \emph{ad hoc} limit to the cost of any particular step. Figure~\ref{fig:GeneticsResultsCost} shows that the accuracy is nearly identical for all the cases, but the approximate chains use fewer evaluations of the true model, reducing costs by more than an order of magnitude for quadratic or Gaussian process approximations (Figure~\ref{fig:GeneticsResults:b}). Local linear approximations show only modest improvements in the cost. Note that when proposals fall outside the support of the prior, the proposal is rejected without running either the true or approximate models; hence even the reference configuration runs the model less than once per MCMC step.

\begin{figure}[htb]
\centering
\includegraphics[scale=.8]{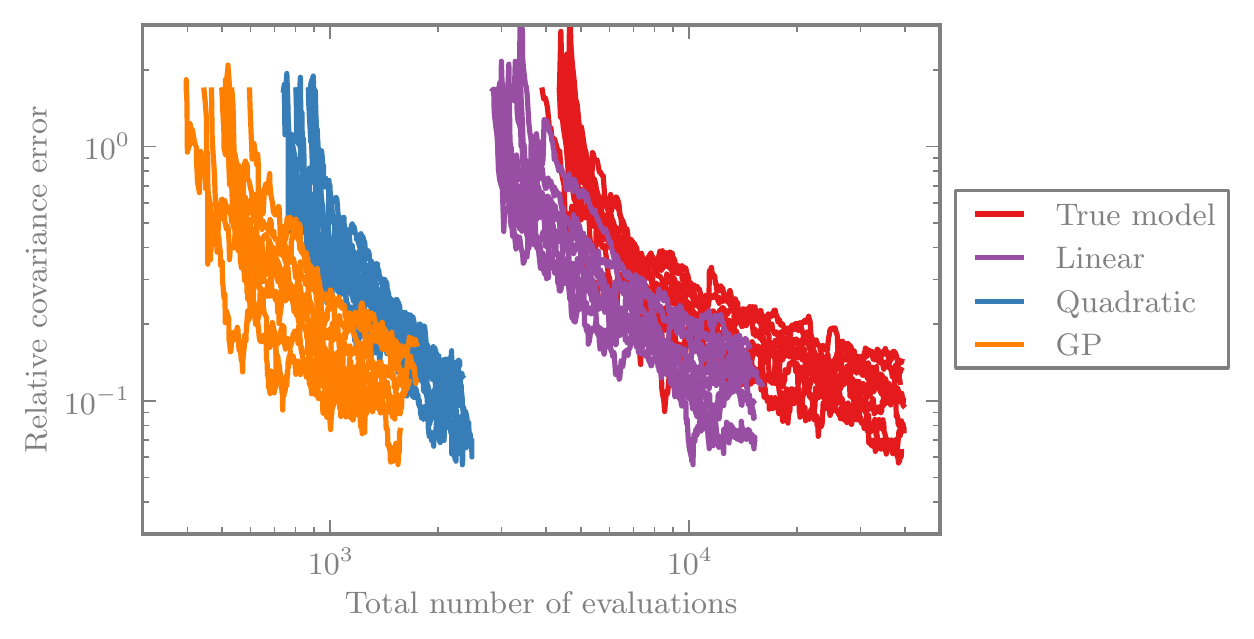}

\caption{Approximate relative covariance errors in the MCMC chains versus their costs, for the genetic toggle switch problem, using several different local approximation strategies. The plot depicts ten independent chains of each type, with the first $10\%$ of each chain removed as burn-in.}
\label{fig:GeneticsResultsCost}
\end{figure}

It is also instructive to plot the accuracy and cost as a function of the number of MCMC steps, as in Figure~\ref{fig:GeneticsResults}. All the accuracy trajectories in Figure~\ref{fig:GeneticsResults:a} lie on top of each other, suggesting that the approximations do not have any discernable impact on the mixing time of the chain. 
Yet Figure~\ref{fig:GeneticsResults:b} shows not only that the approximation strategies yield lower total cost at any given number of MCMC steps, but also that these costs accumulate at a \textit{slower rate} than when the true model is used directly.

\begin{figure}[htb]
\centering
\subfloat[The accuracy of the chains.]{
\label{fig:GeneticsResults:a}
\includegraphics[scale=.8]{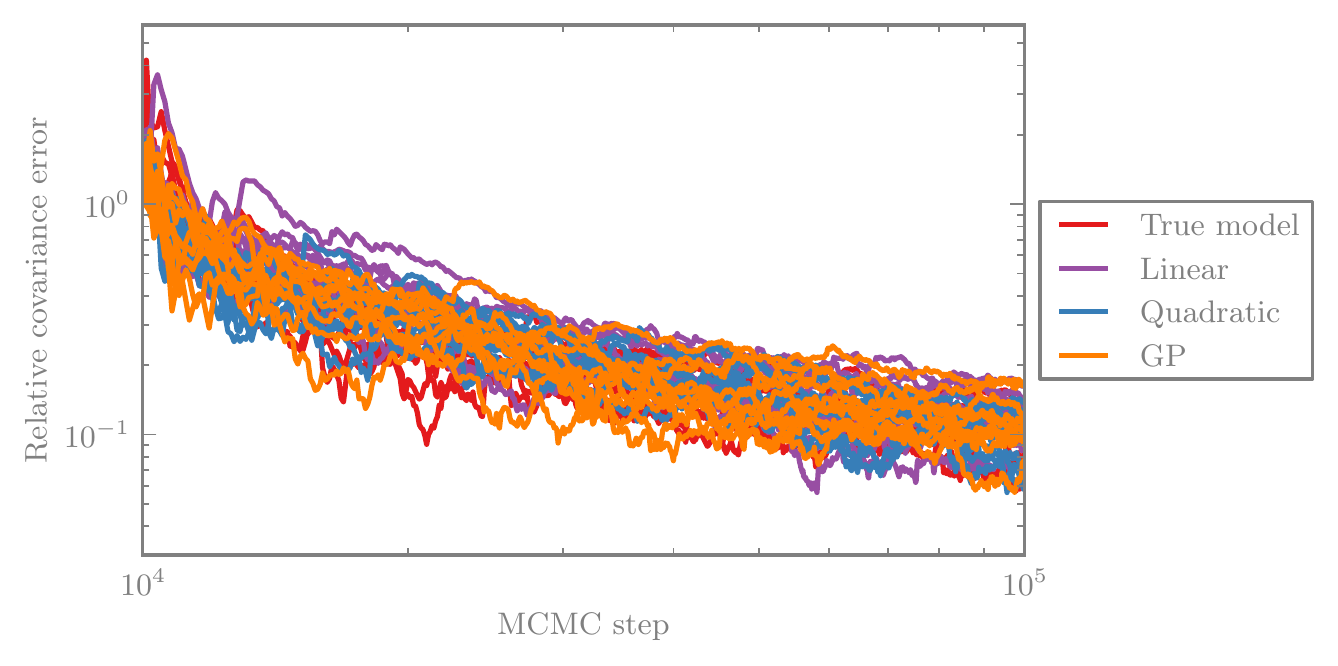}
}\\
\subfloat[The cost of the chains.]{
\label{fig:GeneticsResults:b}
\includegraphics[scale=.8]{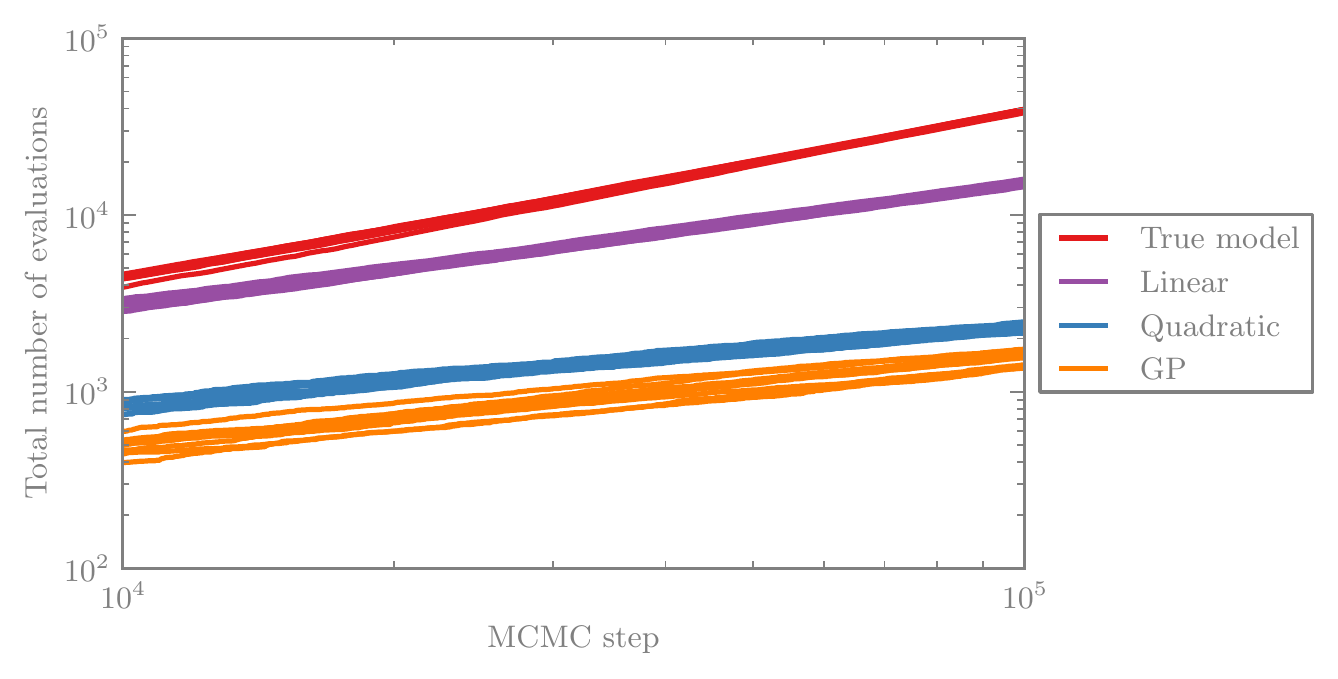}
}
\caption{Approximate relative covariance errors in the MCMC chains and their costs, shown over the length of the MCMC chain, for the genetic toggle switch problem, using several different local approximation strategies. The plot depicts ten independent chains of each type, with the first $10\%$ of each chain removed as burn-in.}
\label{fig:GeneticsResults}
\end{figure}

\subsection{Elliptic PDE inverse problem}
\label{s:elliptic}

We now turn to a canonical inverse problem involving inference of the diffusion coefficient in an elliptic PDE \citep{Dashti2011}. We leave the details of the PDE configuration to Appendix~\ref{apx:pde}; it suffices for our purposes that it is a linear elliptic PDE on a two-dimensional spatial domain, solved with a finite element algorithm at moderate resolution. The diffusion coefficient is defined by six parameters, each endowed with a standard normal prior. Noisy pointwise observations are taken from the solution field of the PDE and are relatively informative, and hence the posterior shifts and concentrates significantly with respect to the prior, as shown in Figure \ref{fig:laplaceMarginals}. We also emphasize that even though the PDE is linear, the forward model---\emph{i.e.}, the map from the parameters to the observed field---is nonlinear and hence the posterior is not Gaussian. We also note that, while the design of effective posterior sampling strategies for functional inverse problems is an enormous and important endeavor \citep{Cotter2013}, our parameterization renders this problem relatively low-dimensional and the simple adaptive Metropolis sampler used to obtain our results mixes well. 

\begin{figure}[htbp]
\centering
\includegraphics[scale=.8]{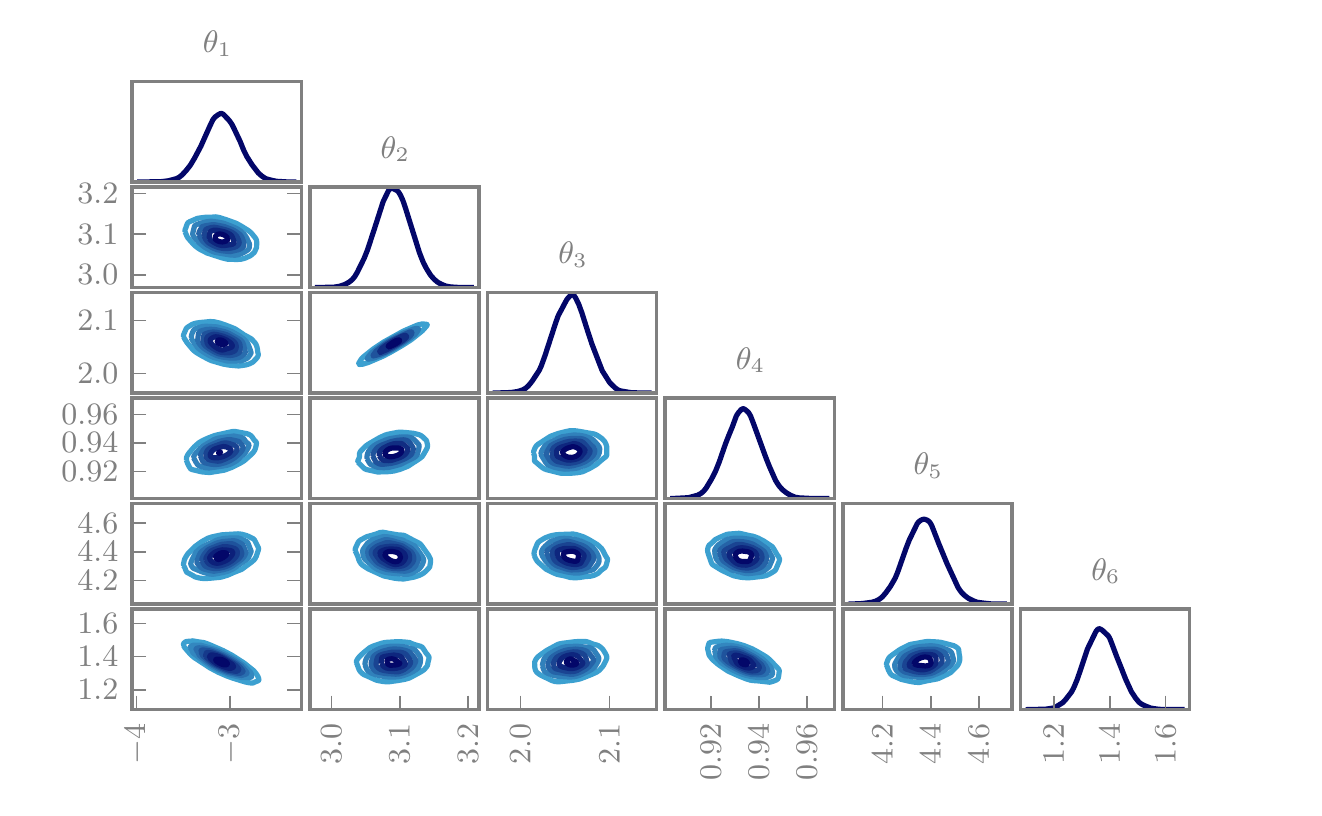}
\caption{One- and two- dimensional posterior marginals of the parameters in the elliptic PDE inverse problem.}
\label{fig:laplaceMarginals}
\end{figure}

Now we evaluate the performance of the various local approximation schemes, using the same experiments as in the previous section; results are summarized in Figure \ref{fig:LaplaceResults}. As in the genetic toggle switch example, the accuracies of all the configurations are nearly indistinguishable, yet the approximate chains demonstrate significantly reduced use of the true forward model. Local linear approximations of the forward model decrease the cost by over an order of magnitude. Both the local quadratic and local GP regressors yield well \emph{over two orders of magnitude reduction} in cost. We suggest that our schemes perform very well in this example both because of the regularity of the likelihood and because the concentration of the posterior limits the domain over which the approximation must be accurate.

\begin{figure}[htb]
\centering
\includegraphics[scale=.8]{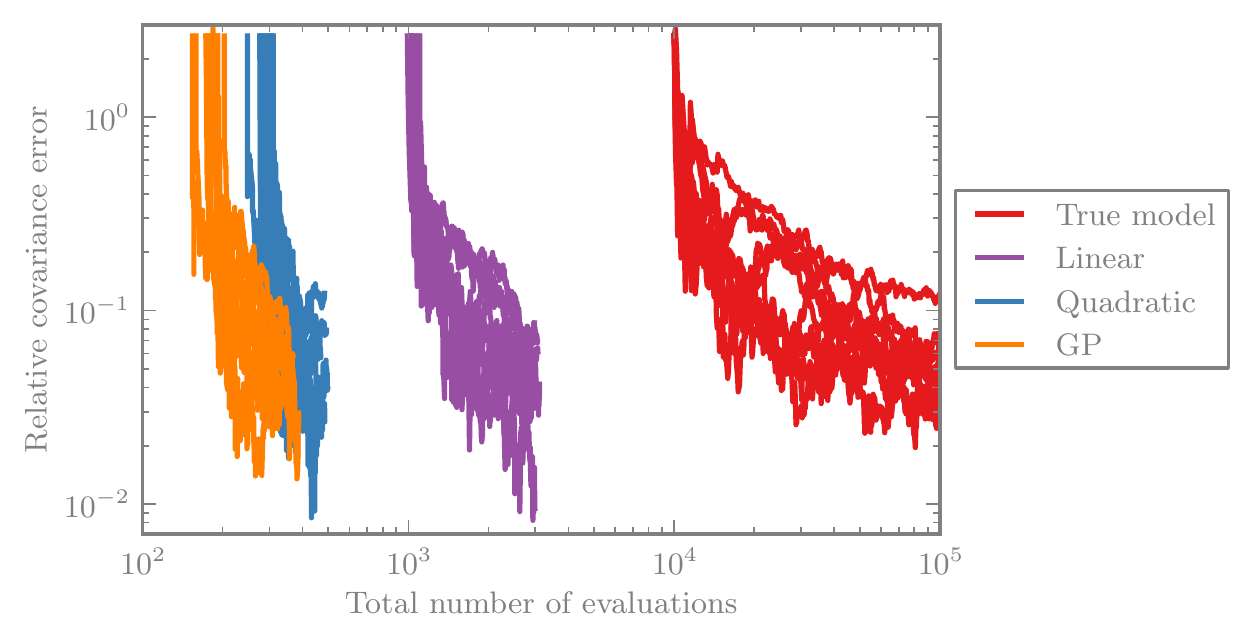}
\caption{Approximate relative covariance errors in the MCMC chains versus their costs, for the elliptic PDE inverse problem, using several different local approximation strategies. The plot depicts ten independent chains of each type, with the first $10\%$ of each chain removed as burn-in.}
\label{fig:LaplaceResults}
\end{figure}


\subsection{Implementation and performance notes}

We have now demonstrated how our approximate MCMC  framework can dramatically reduce the use of the forward model, but we have not yet addressed the performance of our implementation in terms of running time or memory. Although in principle one might worry that the cost of storing the growing sample set $\mathcal{S}$ or of performing the nearest neighbor searches might become challenging, we find that neither is problematic in practice. Storing a few thousand samples, as required in our tests, is trivial on modern machines. Finding nearest neighbors is a hard problem asymptotically with respect to the parameter dimension and size of the sample set, but our sample sets are neither high dimensional nor large. We use an efficient library to perform the nearest neighbor computations, which implements specialized algorithms that can vastly outperform the asymptotic complexity for low-dimensional nearest neighbors \citep{Muja2009}, and we observe that its run time is an insignificant cost. Computing the error indicator is also relatively inexpensive in these settings: for polynomials, each cross-validation sample only requires a low-rank update of the least squares solution; and for Gaussian processes, drawing from the posterior predictive distribution is fast once the GP has been fit.

To investigate the run-time performance, we measured the average wall-clock time needed to construct each chain used in the genetic toggle switch and elliptic PDE examples on a typical desktop: true model (9 and 4 minutes, respectively), linear (4 and 5 minutes), quadratic (5 minutes and 1 hour), Gaussian process (2.4 and 8.5 hours).\footnote{The overhead in computing approximations for the elliptic PDE example is more expensive because the forward model has many more outputs than the genetic toggle switch example.} For quadratic approximations, benchmarking suggests that around 70\% of the run-time was spent computing QR factorizations needed to fit the quadratic surrogates and $<\!2\%$ was spent performing nearest neighbor searches or running the full model. Even though the models take only a small fraction of a second to run, the linear approximation is already competitive in terms of run-time. For sufficiently expensive forward models, the fixed cost of constructing approximations will be offset by the cost of the model evaluations, and real run-times should reflect the strong performance we have demonstrated with problem-invariant metrics. Although Gaussian process approximations showed slightly superior performance in terms of model use, the computational effort required to construct them is much higher, suggesting that they will be most useful for extremely expensive models.


\section{Discussion}
\label{sec:discussion}

We have proposed a new class of MCMC algorithms that construct local surrogates to reduce the cost of Bayesian inference in problems with computationally expensive forward models. These algorithms introduce local approximations of the forward model or log-likelihood into the Metropolis-Hastings kernel and refine these approximations incrementally and infinitely. The resulting Markov chain thus employs a sequence of approximate transition kernels, but asymptotically samples from the exact posterior distribution. We describe variations of the algorithm that employ either local polynomial or Gaussian process approximations, thus spanning two widely-used classes of surrogate models. Gaussian processes appear to provide somewhat superior performance in terms of reducing use of the forward model, but local quadratic models are cheaper to construct; therefore, both seem to be useful options, depending on cost of the true model. In either case, numerical experiments demonstrate significant reductions in the number of forward model evaluations used for posterior sampling in ODE and PDE model problems. 

We do not claim that our algorithm provides minimal error in MCMC estimates given a particular budget of forward model runs; indeed, we expect that problem-specific methods could outperform our strategy in many cases. Instead, we argue that the convergence of the algorithm makes it straightforward to apply to novel problems and to assess the quality of the results. 
The essential reason is that refinement of local approximations is directly tied to the progress of the MCMC chain. As MCMC expends more effort exploring the target distribution, the quality of the approximations increases automatically, via refinement criteria that target problem-independent quantities. 
The cost of constructing the approximations is incurred incrementally and is tuned to correspond to the MCMC sampling effort. 
Although it is not feasible to predict in advance how many MCMC steps or model runs will be needed, difficulty either in exploring the posterior or in approximating the model is typically revealed through non-stationary behavior of the chain. Hence, standard MCMC diagnostics can be used to monitor convergence of the chain \textit{and} the underlying approximation.
This argument is supported by our numerical results, which produce chains whose convergence is largely indistinguishable from that of regular MCMC. Moreover, after initial exploration of the refinement thresholds, numerical results in these examples are obtained without problem-specific tuning.

Our theoretical and numerical results underscore the notion that local regularity in the forward model or log-likelihood should be harnessed for computational efficiency, and that the number of model evaluations needed to approach exact sampling from the posterior can be much smaller than the number of MCMC samples. Although our convergence arguments can be made quantitative, we believe that doing so in a straightforward manner does not capture the greatest strength of our algorithm. Looking at the process described in Example \ref{ex:decrag}, we see that a reasonable start results in a bias bound that decays almost exponentially in the number of likelihood evaluations and that the number of likelihood evaluations will grow approximately logarithmically in the running time of the process. Our general bounds, however, only imply that the bias decays at some rate, which may potentially be quite slow. The discrepancy between these rates comes from the fact that our cross-validation approach attempts to evaluate the likelihood primarily in regions where refinement is important. In situations such as Example \ref{ex:decrag}, these well-chosen likelihood evaluations give a much better estimate than would be obtained from points chosen according to the posterior distribution; in other cases, they seem to be similar. A more general theory would need to avoid the problems that arise in Example \ref{ExDecRateBeta} and similar constructions.

There remains significant room to develop other algorithms within this framework. A wide variety of local approximations have theoretical convergence properties similar to those exploited here, offering the opportunity to explore other families of approximations, different weight functions and bandwidths, or variable model order, \emph{cf.}\ \citep{Cleveland1996, Gramacy2013}. Other variations include constructing surrogates by sharing $\mathcal{S}$ across parallel MCMC chains; using any available derivative information from the forward model to help construct local approximations; or using local approximations as corrections to global surrogates, creating hybrid strategies that should combine the fast convergence of global approximations with the asymptotic exactness of our construction \citep{Chakraborty2013}. It should also be possible to extend our use of local approximations to other varieties of MCMC; of particular interest are derivative-based methods such as Metropolis-adjusted Langevin (MALA) or Hybrid Monte Carlo (HMC), where the easy availability of derivatives from our local approximations can dramatically impact their feasibility \citep{Rasmussen2003}. Several of these variations are explored in \cite{Conrad2014a}. Finally, further work may reveal connections between the present strategy and other methods for intractable likelihoods, such as pseudo-marginal MCMC, or with data assimilation techniques for expensive models \citep{law2015data}. 

\section*{Acknowledgments}
P.\ Conrad and Y.\ Marzouk acknowledge support from the Scientific Discovery through Advanced Computing (SciDAC) program funded by the US Department of Energy, Office of Science, Advanced Scientific Computing Research under award number DE-SC0007099. N.\ Pillai is partially supported by the grant ONR 14-0001. He thanks Dr.\ Pedja Neskovic for his interest in this work. Aaron Smith was supported by a grant from the Natural Sciences and Engineering Research Council of Canada.



\bibliographystyle{chicago} 
\bibliography{library,library_manual} 


 \newpage
\appendix

\section{Local polynomial regression}
\label{sec:polyDetails}

Here we provide additional detail about the polynomial regression scheme described in Section~\ref{s:polyapprox}. We consider the quadratic case, as the linear case is a simple restriction thereof. For each component $f_j$ of $\mathbf{f}$, the quadratic regressor is of the form 
\begin{equation*}
\tilde{{f}}_j(\hat{\theta}) := a_j + b_j^T \hat{\theta} + \frac{1}{2}\hat{\theta}^T H_j \hat{\theta},
\end{equation*}
where $a_j \in \mathbb{R}$ is a constant term, $b_j \in \mathbb{R}^d$ is a linear term, and $H_j \in \mathbb{R}^{d \times d}$ is a symmetric Hessian matrix. Note that $a_j$, $b_j$, and $H_j$ collectively contain $M = (d+2)(d+1)/2$ independent entries for each $j$. The coordinates $\hat{\theta} \in \mathbb{R}^d$ are obtained by shifting and scaling the original parameters $\theta$ as follows. Recall that the local regression scheme uses $N$ samples $\{\theta^1, \ldots, \theta^N\}$ drawn from the ball of radius $R$ centered on the point of interest $\theta$, along with the corresponding model evaluations $y_j^i = f_j(\theta^i)$.\footnote{To avoid any ambiguities, this appendix departs from the rest of the narrative by using a superscript to index samples and a subscript to index coordinates.} We assume that the components of $\theta$ have already been scaled so that they are of comparable magnitudes, then define $\hat{\theta}^i = (\theta^i - \theta)/R$, so that the transformed samples are centered at zero and have maximum radius one. Writing the error bounds as in (\ref{e:bounds}) requires this rescaling along with the $1/2$ in the form of the regressor above \citep{Conn2009}. 

Next, construct the diagonal weight matrix $W = \mathrm{diag}(w^1, \ldots, w^N)$ using the sample weights in (\ref{eqn:tricubeWeights}), where we have $R=1$ because of the rescaling. Then compute the $N$-by-$M$ basis matrix $\Phi$:

\begin{equation*}
\Phi = 
\begin{pmatrix}
1 & \hat{\theta}_1^1 & \cdots & \hat{\theta}_d^1 & \frac{1}{2} \! \left (\hat{\theta}_1^1 \right )^2 & \cdots & \frac{1}{2} \! \left (\hat{\theta}_d^1 \right )^2 & \hat{\theta}_1^1\hat{\theta}_2^1 & \cdots & \hat{\theta}_{d-1}^1 \hat{\theta}_d^1\\
\vdots & & & & & & & & & \vdots \\
1 & \hat{\theta}_1^N & \cdots & \hat{\theta}_d^N & \frac{1}{2} \! \left (\hat{\theta}_1^N \right )^2 & \cdots & \frac{1}{2} \! \left (\hat{\theta}_d^N \right )^2 & \hat{\theta}_1^N\hat{\theta}_2^N & \cdots & \hat{\theta}_{d-1}^N \hat{\theta}_d^N
\end{pmatrix}
\end{equation*}
where we ensure that $N > M$. Finally, solve the $n$ least squares problems,
\begin{equation}
\Phi^T W \Phi Z = \Phi^T W Y,
\label{e:lsq}
\end{equation}
where each column of the $N$-by-$n$ matrix $Y$ contains the samples $\left ( y_j^1,  \ldots, y_j^N  \right )^T$, $j=1, \ldots, n$. Each column $z_j$ of $Z \in \mathbb{R}^{M \times n}$ contains the desired regression coefficients for output $j$,
 \begin{equation}
 z_j^T = \begin{pmatrix}
 a_j & b_j^T  & (H_j)_{1,1} & \cdots & (H_j)_{d,d} & (H_j)_{1,2} & \cdots & (H_j)_{d-1,d}
 \end{pmatrix}.
 \end{equation}
The least squares problem may be solved in a numerically stable fashion using a QR factorization of $W \Phi Z$, which may be computed once and reused for all $n$ least squares problems. The cross-validation fit omitting sample $i$ simply removes row $i$ from both sides of (\ref{e:lsq}). These least squares problems can be solved efficiently with a low-rank update of the QR factorization of the full least squares problem, rather than recomputing the QR factors from scratch \citep{Hammarling2008}. 
\section{Detailed theoretical results and proofs of theorems}
\label{sec:theoryappendix}
\subsection{Auxiliary notation}
We now define some useful auxillary objects. For a fixed  finite set $\mathcal{S} \subset \Theta$, we consider the stochastic process defined by Algorithm \ref{alg:algOverview} with $\mathcal{S}_{1} = \mathcal{S}$ and lines 11--20 and 22 removed. This process is essentially the original algorithm with all approximations based on a single set of points $\mathcal{S}$ and no refinements. Since there are no refinements, this process is in fact a Metropolis-Hastings Markov chain, and we write  $K_{\mathcal{S}}$ for its transition kernel. For all measurable sets $U \subset \Theta$, this kernel can be written as $K_{\mathcal{S}}(x,U) = r_{\mathcal{S}}(x) \delta_{x}(U) + (1 - r_{\mathcal{S}}(x)) \int_{y \in U} p_{\mathcal{S}}(x,y)dy$ for some $0 \leq r_{\mathcal{S}}(x) \leq 1$ and  density $p_{\mathcal{S}}(x,y)$. We denote by $\alpha_{\mathcal{S}}(x,y)$ the acceptance probability of $K_{\mathcal{S}}$. 

We introduce another important piece of notation before giving our results.  Let $\{ Z_{t} \}_{t \in \mathbb{N}}$ be a (generally non-Markovian) stochastic process on some state space $\Omega$. We say that a sequence of (generally random, dependent) kernels $\{Q_{t} \}_{t \in \mathbb{N}}$ is \textit{adapted} to $\{ Z_{t} \}_{t \in \mathbb{N}}$ if there exists an auxillary process $\{ A_{t} \}_{t \in \mathbb{N}}$ so that:
\begin{itemize}
\item  $\{ (Z_{t}, A_{t} ) \}_{t \in \mathbb{N}}$ is a Markov chain,
\item $Q_{t}$ is $\sigma(A_{t})$-measurable, and
\item $\P[Z_{t+1} \in \cdot | Z_{t}, A_{t}] = Q_{t}(Z_{t}, \cdot)$.
\end{itemize} 

Let $\{ X_{t}, \mathcal{S}_{t} \}_{t \in \mathbb{N}}$ be a sequence evolving according to the stochastic process defined by Algorithm \ref{alg:algOverview} and define the following associated sequence of kernels:
\be 
\tilde{K}_{t}(x, A) \equiv \P[X_{t+1} \in A | \{X_{s}\}_{1 \leq s < t}, X_{t} = x, \{\mathcal{S}_{s} \}_{1 \leq s \leq t}].
\ee 
The sequence of kernels $\{ \tilde{K}_{t} \}_{t \in \mathbb{N}}$ is adapted to $\{ X_{t} \}_{t \in \mathbb{N}}$, with $\{ \mathcal{S} \}_{t \in \mathbb{N}}$ as the auxillary process. For any fixed $t$, one can sample from $\tilde{K}_{t}(x, \cdot)$ by first drawing a proposal $y$ from $L(x, \cdot)$ and then accepting with probability 
\be \label{EqAlphDef} 
\tilde{\alpha}_{t}(x, y) \equiv c_{1} \alpha_{\mathcal{S}_{t}}(x,y) + c_{2} \alpha_{\mathcal{S}_{t} \cup \{ (x,f(x)) \}}(x,y) + c_{3} \alpha_{\mathcal{S}_{t} \cup \{ (y,f(y)) \}}(x,y),
\ee  
where $c_1, c_2, c_3$ are some positive constants that depend on $x,y,\beta_{t}$ and $\gamma_{t}$ and satisfy the identity $c_{1} + c_{2} + c_{3} = 1$.

\subsection{Book-keeping result}
The following result will be used repeatedly in our ergodicity arguments.
\begin{theorem} [Approximate Ergodicity of Adaptive Chains] \label{CorMainProbIneq}
Fix a kernel $K$ with stationary distribution $\pi$ on state space $\mathcal{X}$ and let $\{Y_{t}\}_{t \geq 0}$ evolve according to $K$. Assume
\be \label{IneqBasicResPos3}
\| K^{t}(x,\cdot) - \pi \|_{\TV} \leq C_{x} (1 - \alpha)^{t}
\ee  
for some $0 < \alpha \leq 1$, $\{ C_{x} \}_{x \in \mathcal{X}}$ and all $t \in \mathbb{N}$.

Let $\{ K_{t} \}_{t \in \mathbb{N}}$ be a sequence of kernels adapted to some stochastic process $\{ X_{t} \}_{t \in \mathbb{N}}$, with auxillary process $\{A_{t} \}_{t \in \mathbb{N}}$. Also fix a Lyapunov function $V$ and constants $0 < a, \delta, \epsilon < 1$, $0 \leq b< \infty$ and $0 \leq B < \frac{2b}{ a \epsilon}$. Assume that there exists a non-random time $\mathcal{T} = \mathcal{T}_{\epsilon, \delta}$ and a $\sigma \left( \{(X_{s},A_{s}) \}_{s \in \mathbb{N}}^{\mathcal{T}} \right)$-measurable event $\mathcal{F}$ so that $\P[\mathcal{F}] > 1 - \epsilon$, 
\be 
\E[ V(X_{\mathcal{T}}) \textbf{1}_{\mathcal{F}}] &< \infty,\label{IneqBasicResPos1} \\
\sup_{t > \mathcal{T}} \sup_{x \, : \, V(x) < B } \| K_{t}(x, \cdot) - K(x,\cdot) \|_{\TV} &< \delta + \textbf{1}_{\mathcal{F}^{c}},\label{IneqBasicResPos2}
\ee  
and the following inequalities are satisfied for all $t > \mathcal{T}$:

\be \label{IneqBasicResPos4}
\E[V(X_{t+1}) \textbf{1}_{\mathcal{F}}  | X_{t}=x, A_{t}] &\leq (1 - a) V(x) + b \\
\E[V(Y_{t+1}) | Y_{t}=y] &\leq (1 - a) V(y) + b.
\ee  

Then
\be 
\limsup_{T \rightarrow \infty} \| \mathcal{L}(X_{T}) - \pi \|_{\TV} &\leq   3\epsilon  + \delta \frac{ \log \left( \frac{e \delta}{\mathcal{C} \log(1-\alpha)} \right) }{\log(1 - \alpha)} +  \frac{4b}{a B} \lceil \frac{ \log \left( \frac{\delta}{\mathcal{C} \log(1-\alpha)} \right) }{\log(1 - \alpha)} + 1 \rceil,
\ee 
where $\mathcal{C} = \mathcal{C}(\epsilon) \equiv \sup \{ C_{x} \, : \, V(x) \leq \frac{2b}{\epsilon a}  \} $. 
\end{theorem}

\begin{proof}
Assume WLOG that $\mathcal{T} = 0$, fix $\gamma > 0$ and fix $\frac{\log \frac{b}{a ( \max(\E[V(X_{0}) \textbf{1}_{\mathcal{F}}], \pi(V)) +1)}}{\log(1-a)} \leq S < T$. Let $\{Y_{t} \}_{t \geq S}$, $\{ Z_{t} \}_{t \geq S}$ be Markov chains evolving according to the kernel $K$ and starting at time $S$, with $Y_{S} = X_{S}$ and $Z_{S}$ distributed according to $\pi$ . By inequality \eqref{IneqBasicResPos3}, it is possible to couple $\{Y_{t} \}_{S \leq t \leq T}$, $\{Z_{t} \}_{S \leq t \leq T}$ so that 
\be  \label{MidBasicProof1}
\P[Y_{T} \neq Z_{T} | X_{S}] \leq C_{X_{S}}( 1 - \alpha)^{T-S} + \gamma.
\ee 
By inequality \eqref{IneqBasicResPos2} and a union bound over $S \leq t < T$, it is possible to couple $\{ X_{t} \}_{S \leq t \leq T}$, $\{ Y_{t} \}_{S \leq t \leq T}$ so that
\be \label{MidBasicProof2}
\P[X_{T} \neq Y_{T}] \leq \delta (T-S) + \P[\mathcal{F}^{c}]  + \P[\max_{S \leq t \leq T} (\max(V(X_{t}), V(Y_{t}))) > B] + \gamma.
\ee 
By inequalities \eqref{IneqBasicResPos1} and \eqref{IneqBasicResPos4},
\be 
\E[V(X_{S}) \textbf{1}_{\mathcal{F}} |  X_{0}, A_{0} ] \leq \E[V(X_{0}) \textbf{1}_{\mathcal{F}}] (1 - a)^{S} + \frac{b}{a} \leq \frac{2b}{a},
\ee 
and so by Markov's inequality, 
\be \label{MidBasicProof3}
\P \big[ \{ V(X_{S}) > \frac{2b}{ a \epsilon}\} \cap \mathcal{F} \big] \leq \epsilon.
\ee 
By the same calculations, 
\be  \label{MidBasicProof3a}
\P[ \{\max_{S \leq t \leq T} ( \max(V(X_{t}), V(Y_{t}))) > B  \} \cap \mathcal{F}] \leq (T-S + 1) \frac{4b}{a B}. 
\ee 
Couple $\{ Y_{t} \}_{S \leq t \leq T}$ to $\{ X_{t} \}_{S \leq t \leq T}$ so as to satisfy inequality \eqref{MidBasicProof2}, and then couple $\{Z_{t} \}_{S \leq t \leq T}$ to $\{Y_{t} \}_{S \leq t \leq T}$ so as to satisfy inequality \eqref{MidBasicProof1}. It is possible to combine these two couplings of pairs of processes into a coupling of all three processes by the standard `gluing lemma' (see \textit{e.g.}, Chapter 1 of \cite{vill09}). Combining inequalities \eqref{MidBasicProof1}, \eqref{MidBasicProof2}, \eqref{MidBasicProof3}, and \eqref{MidBasicProof3a}, we have
\be 
\| \mathcal{L}(X_{T}) - \pi \|_{\TV} &\leq \P[X_{T} \neq Y_{T}] + \P[Y_{T} \neq Z_{T}]\\
&\leq \P[  X_{T} \neq Y_{T} ]  + \E[ \textbf{1}_{Y_{T} \neq Z_{T}} \textbf{1}_{ V(X_{S}) > B} \textbf{1}_{\mathcal{F}}] + \E[ \textbf{1}_{Y_{T} \neq Z_{T}} \textbf{1}_{V(X_{S}) \leq  B}] + \P[\mathcal{F}^{c}] \\
&\leq \delta (T-S) + 3 \epsilon + (T-S + 1) \frac{4b}{a B} + 2 \gamma + \mathcal{C} (1 - \alpha)^{T-S}. 
\ee  
Approximately optimizing over $ S < T$ by choosing  $S' = T - \lceil \frac{ \log \left( \frac{\delta}{\mathcal{C} \log(1-\alpha)} \right) }{\log(1 - \alpha)} \rceil$ for $T$ large, we conclude
\be
\limsup_{T \rightarrow \infty} \| \mathcal{L}(X_{T}) - \pi \|_{\TV} &\leq \limsup_{T \rightarrow \infty}  \left( \delta (T-S') + 3 \epsilon + (T-S + 1) \frac{4b}{a B} + 2 \gamma + \mathcal{C} (1 - \alpha)^{T-S'} \right) \\
&\leq 3\epsilon + 2 \gamma + \delta \frac{ \log \left( \frac{\delta}{\mathcal{C} \log(1-\alpha)} \right) }{\log(1 - \alpha)} + \frac{\delta}{\log(1 - \alpha)} +  \frac{4b}{a B} \lceil \frac{ \log \left( \frac{\delta}{\mathcal{C} \log(1-\alpha)} \right) }{\log(1 - \alpha)} + 1 \rceil.
\ee 
Since this holds for all $\gamma > 0$, the proof is finished.\end{proof}

\begin{remark}
In the adaptive MCMC literature, similar results are often stated in terms of a \textit{diminishing adaptation} condition (this roughly corresponds to inequality \eqref{IneqBasicResPos2}) and a \textit{containment} condition (this roughly corresponds to inequalities \eqref{IneqBasicResPos1} and \eqref{IneqBasicResPos4}). These phrases were introduced in \cite{RobRos07}, and there is now a large literature with many sophisticated variants; see, e.g., \cite{Fort2012} for related results that also give LLNs and CLTs under similar conditions. We included our result because its proof is very short, and because checking these simple conditions is easier than checking the more general conditions in the existing literature. \end{remark}

\subsection{Good sets and monotonicity}
We give some notation that will be used in the proofs of Theorems \ref{ThmCompSup} and  \ref{ThmErgGaussEnv}. Fix $0 \leq c, r,R \leq \infty$. For $0 < \ell < \infty$ and $x \in \mathbb{R}^{d}$, denote by $\mathcal{B}_{\ell}(x)$ the ball of radius $\ell$ around $x$. Say that a finite set $\mathcal{S} \subset \Theta \subset \mathbb{R}^{d}$ is $(c, r, R)$-\textit{good} with respect to a set $\mathcal{A} \subset \Theta$ if it satisfies:
\begin{enumerate}
\item $\sup_{ x \in \mathcal{A}, \|x\| \leq r} \min_{y \in \mathcal{S}} \| x - y \| \leq c$. 
\item For all $x \in \mathcal{A}$ with $\| x \| > R$, we have that $| \mathcal{S} \cap \mathcal{B}_{\frac{1}{2} \| x \|}(x) | \geq N$. 
\end{enumerate}
We say that it is $(c, r, R)$-\textit{good} if it is $(c, r, R)$-\textit{good} with respect to $\Theta$ itself. The first condition will imply that the approximation $p_{\mathcal{S}}(x)$ is quite good for $x$ close to the origin. The second condition gives an extremely weak notion of `locality'; it implies the points we use to construct a `local' polynomial approximation around $x$ do not remain near the origin when $\| x \|$ itself is very far from the origin.  We observe that our definition is monotone in various parameters:
\begin{itemize}
\item If $\mathcal{S}$ is $(c,r,R)$-good, then it is also $(c',r',R')$-good for all $c' \geq c$, $r' \leq r$ and $R' \geq R$.
\item  If $\mathcal{S}$ is $(c,r,R)$-good, then $\mathcal{S} \cup \mathcal{S}'$ is also $(c,r,R)$-good for any finite set $\mathcal{S}' \subset \Theta$. 
\item If $\mathcal{S}$ is $(\infty,0,R)$-good and $(c,r,\infty)$-good, it is also $(c,r,R)$-good.
\end{itemize}
Our arguments will involve showing that, for any finite $(c,r,R)$, the sets $\{ \mathcal{S}_{t} \}_{t \geq 0}$ are eventually $(c,r,R)$-good.

\subsection{Proof of Theorem \ref{ThmCompSup}, ergodicity in the compact case} \label{SecErgPfComp}
In this section we give the proof of Theorem \ref{ThmCompSup}. Note that some statements are made in slightly greater generality than necessary, as they will be reused in the proof of Theorem \ref{ThmErgGaussEnv}.
\begin{lemma} [Convergence of Kernels] \label{LemGridRefComp}
Let the assumptions stated in the statement of  Theorem \ref{ThmCompSup} hold. For all $\delta > 0$, there exists a stopping time $\tau = \tau(\delta)$ with respect to $\{ \mathcal{S}_{t} \}_{t \in \mathbb{N}}$ \footnote{Throughout the note, for any stochastic process $\{Z_{t}\}_{t \geq 0}$, we use the phrase ``$\tau$ is a stopping time with respect to $\{ Z_{t} \}_{t \geq 0}$" as shorthand for ``$\tau$ is a stopping time with respect to the filtration $\mathcal{F}_{t}$ given by $\mathcal{F}_{t} = \sigma (\{Z_{s}\}_{0 \leq s \leq t})$." }so that
\be \label{IneqFirstGoalCmt} 
\sup_{t > \tau} \, \sup_{x \in \Theta} \| K_{\infty}(x,\cdot) - \tilde{K}_{t}(x,\cdot) \|_{\TV} < \delta 
\ee 
and so that $\P[\tau < \infty] = 1$.
\end{lemma}
\begin{proof}
Fix $R \in \mathbb{R}$ so that $\Theta \subset \mathcal{B}_{R}(0)$. By results in \citep{Conn2009},\footnote{The required result is a combination of Theorems 3.14 and 3.16, as discussed in the text after the proof of Theorem 3.16 of \citep{Conn2009}.} for any $\lambda, \alpha > 0$, there exists a constant $c = c(\alpha, \lambda) > 0$ so that $\sup_{\theta \in \Theta}|p_{\mathcal{S}}(\theta) - p(\theta | \d) | < \alpha$ if $\mathcal{S}$ is $\lambda$-poised and $(c,R,R)$-good. Set $c = c(\delta, \lambda)$ and define $\tau = \inf \{ t \, : \, \mathcal{S}_{t} \text{ is } (c,R,R)- \text{good} \}$. By definition, this is a stopping time with respect to $\{ \mathcal{S}_{t} \}_{t \in \mathbb{N}}$ that satisfies inequality \eqref{IneqFirstGoalCmt}; we now check that $\P[\tau < \infty] = 1$. \par

By the assumption that $\ell(x,y)$ is bounded away from 0, there exist $\epsilon>0$ and measures $\mu$, $\{r_{x}\}_{x \in \Theta}$ so that 
\be \label{EqRepMinCmt}
L(x, \cdot) = \epsilon \mu(\cdot) + (1 - \epsilon) r_{x}(\cdot).
\ee 

Let $\{A_{i} \}_{i \in \mathbb{N}}$ and $\{ B_{i} \}_{i \in \mathbb{N}}$ be two sequences of i.i.d. Bernoulli random variables, with success probabilities $\epsilon$ and $\beta$ respectively. Let $\tau_{0} = \inf \{ t \, : \, X_{t} \in \Theta \}$ and define inductively $\tau_{i+1} = \inf \{ t > \tau_{i} + 1 \, : \, X_{t} \in \Theta \}$. By equality \eqref{EqRepMinCmt}, it is possible to couple the sequences $\{ X_{t} \}_{t \in \mathbb{N}}, \{A_{i} \}_{i \in \mathbb{N}}$ so that 
\be \label{EqDecompPropCmt}
\P[L_{\tau_{i}} \in \cdot | \tau_{i}, X_{\tau_{i}}, A_{i} = 1] &=   \mu(\cdot)  \\
\P[L_{\tau_{i}} \in \cdot | \tau_{i}, X_{\tau_{i}}, A_{i} = 0] &=   r_{X_{\tau_{i}}}(\cdot).
\ee 

We can further couple $\{B_{i}\}_{i \in \mathbb{N}}$ to these sequences by using $B_{i}$ for the random variable in step 12 of Algorithm \ref{alg:algOverview} at time $\tau_{i}$. That is, when running  Algorithm \ref{alg:algOverview}, we would run the subroutine $\mathrm{RefineNear}$ in step 13 of the algorithm at time $t = \tau_{i}$ if $B_{i} = 1$, and we would not run that subroutine in that step at that time if $B_{i} = 0$. Define $I = \{ i \in \mathbb{N} \, : \, A_{i} = B_{i} = 1 \}$. Under this coupling of $\{A_{i} \}_{i \in \mathbb{N}}, \{B_{i} \}_{i \in \mathbb{N}}$, and $\{X_{t} \}_{t \in \mathbb{N}}$,
\be 
\{ L_{\tau_{i}} \}_{i \in I, \, \tau_{i} < t }  \subset \mathcal{S}_{t}.
\ee 
Furthermore, $\{L_{\tau_{i}} \}_{i \in I, \, i \leq N }$ is an i.i.d sequence of $N$ draws from $\mu$ and $\P[\tau_{i} < \infty] = 1$ for all $i$. Let $\mathcal{E}_{j}$ be the event that $\{ L_{\tau_{i}} \}_{i \leq j}$ is $(c,R,R)$-good. We have $\tau \leq \tau_{\inf \{ j \, : \, \mathcal{E}_{j} \text{ holds} \}}$. By independence of the sequence $\{ L_{\tau_{i}} \}_{i \in \mathbb{N}}$, we obtain
\be 
\P[\tau < \infty] \geq \liminf_{j \rightarrow \infty} \P[\mathcal{E}_{j}] = 1.
\ee 
This completes the proof of the Lemma.
\end{proof}

\begin{remark}
 We mention briefly that this lemma can also be used to obtain a quantitative bound on the asymptotic rate of convergence of the bias of our algorithm.  

Observe that $\tau$ as defined in the proof of Lemma \ref{LemGridRefComp} is stochastically dominated by an exponential distribution with mean $O(-d c^{-d} \log(c) )$ as long as both $\ell(x,\cdot)$ and $p(\cdot | \d)$ are bounded below. This gives a rather poor bound on the amount of time it takes for inequality \eqref{IneqFirstGoal} to hold. Inequality \eqref{IneqFirstGoal}, together with standard `perturbation' bounds relating the distance between transition kernels and the distance between their stationary distributions, imply a quantitative bound on the asymptotic rate of convergence of the bias of our algorithm. An example of such a perturbation bound may be found by applying Theorem 1 of \citep{Korattikara2013}, which does not in fact rely on time-homogeneity, to a subsequence of the stochastic process generated by our algorithm. Unfortunately, the resulting bound is rather poor, and does not seem to reflect our algorithm's actual performance. 
\end{remark}

We now prove Theorem \ref{ThmCompSup}: 

\begin{proof}
It is sufficient to show that, for all $\epsilon, \delta > 0$ sufficiently small, the conditions of Theorem \ref{CorMainProbIneq} can be satisfied. We now set the constants and functions associated with Theorem \ref{CorMainProbIneq}; we begin by choosing $C_{x} \equiv V(x) \equiv b =a= 1$, setting $\alpha = \frac{\inf_{x,y \in \Theta} \ell(x,y) \inf_{\theta \in \Theta} p(\theta | d)}{\sup_{\theta \in \Theta} p(\theta | d)}$, and setting $B = \infty$. \par
By the minorization condition, inequality \eqref{IneqBasicResPos3} is satisfied for this value of $\alpha$; by the assumption that $\ell(x, y), p(\theta | d)$ are bounded away from 0 and infinity, we also have $\alpha > 0$. Next, for all $\delta > 0$, Lemma \ref{LemGridRefComp} implies that implies that $\sup_{x} \| K(x,\cdot) - \tilde{K}(x,\cdot) \|_{\TV} < \delta$ for all times $t$ greater than some a.s. finite random time $\tau = \tau(\delta)$ that is a stopping time with respect to $\{ \mathcal{S}_{t} \}_{t \in \mathbb{N}}$. Choosing $\mathcal{T} = \mathcal{T}_{\epsilon, \delta}$ to be the smallest integer so that $\P[\tau(\delta) > \mathcal{T}] \leq 1 - \epsilon$ and setting $\mathcal{F} = \{\tau \leq \mathcal{T} \}$, this means that inequality \eqref{IneqBasicResPos2} is satisfied. Inequalities \eqref{IneqBasicResPos4} and \eqref{IneqBasicResPos1} are trivially satisfied given our choice of $V,a,b$. Applying Theorem \ref{CorMainProbIneq} with this choice of $V, \alpha,  a, b, \mathcal{T}$, we have for all $\epsilon, \delta > 0$ that 
\be 
\limsup_{T \rightarrow \infty} \| \mathcal{L}(X_{T}) - \pi \|_{\TV} \leq 3\epsilon  + \delta \frac{ \log \left( \frac{e \delta}{\mathcal{C} \log(1-\alpha)} \right) }{\log(1 - \alpha)}.
\ee  
Letting $\delta$ go to 0 and then $\epsilon$ go to 0 completes the proof.
\end{proof}


\subsection{Proof of Theorem \ref{ThmErgGaussEnv}, ergodicity in the non-compact case}
In this section, we prove Theorem \ref{ThmErgGaussEnv}. The argument is similar to that of Theorem \ref{ThmCompSup}, but we must show the following to ensure that the sampler does not behave too badly when it is far from the posterior mode: 
\begin{enumerate}
\item $\mathcal{S}_{t}$ is $(\infty,0,R)$-good after some almost-surely finite random time $\tau$; see Lemma \ref{LemInfApprIgnoreZero}.
\item The kernel $\tilde{K}_{t}$ satisfies a drift condition if $\mathcal{S}_{t}$ is $(\infty,0,R)$-good; see Lemmas \ref{LemmAppAtInf} 
and \ref{LemmaInfDrift}.
\item This drift condition implies that the chain $X_{t}$ spends most of its time in a compact subset of $\Theta$; see Lemma \ref{LemmaInfManyRets}.
\end{enumerate}
\begin{remark}
The Gaussian envelope condition (see Assumption \ref{ass:GE}) is used only to show the second step in the above proof strategy, which in turn is used to satisfy condition \eqref{IneqBasicResPos4} of Theorem \ref{CorMainProbIneq}. It can be replaced by any assumption on the target density for which $\mathcal{S}$ being $(\infty,0,R)$-good for some $R < \infty$ implies that $\tilde{K}_{\mathcal{S}}$ satisfies a drift condition of the form given by inequality \eqref{IneqBasicResPos4}. 
\end{remark}

We begin by showing, roughly, that for any $R > 0$, $\mathcal{S}_{t}$ is eventually $(\infty,0,R)$-good: 
\begin{lemma} [Approximations At Infinity Ignore Compact Sets] \label{LemInfApprIgnoreZero}
Fix any $\mathcal{X} > 0$ and any $k \geq 2$ and define
\be 
\tau_{\mathcal{X}}^{(k)} = \sup \Big \{ t \, : \,  \| L_{t} \| > k \mathcal{X}, \, \| L_{t} \| - R_{t} < \mathcal{X} \Big \}.
\ee 
Then 
\be 
\P[\{\mathrm{There\, exists} \, k < \infty, \, \,\text{s.t.} \, \, \tau_{\mathcal{X}}^{(k)} < \infty\} ] = 1.
\ee 
\end{lemma}
\begin{proof}
Fix $N \in \mathbb{N}$, $\delta > 0$ and $0 < r_{1} < r_{2} < \infty$. For $ 0 < \ell < \infty$, denote by $\partial \mathcal{B}_{\ell}(0)$ the sphere of radius $\ell$. Fix a finite covering $\{ P_{i} \}$ of $\partial \mathcal{B}_{\frac{r_{1}+ r_{2}}{2}}(0)$ with the property that, for any  $ x \in \partial \mathcal{B}_{\frac{r_{1}+ r_{2}}{2}}(0)$, there exists at least one $i$ so that $P_{i} \subset \mathcal{B}_{\delta}(x)$. For $k \in \mathbb{N}$, define a thickening of $P_{i}$ by:
\be 
\mathcal{P}^{(k)}_{i} = \left\{ x \, : \, \frac{r_{1} + r_{2}}{2}\frac{x}{\|x\|} \in P_{i}, \, \, \frac{r_{1} + r_{2}}{2} + (k-1) \frac{r_{2} - r_{1}}{2} \leq \| x \| \leq \frac{r_{1} + r_{2}}{2} + k \frac{r_{2} - r_{1}}{2}  \right\}.
\ee 
We will show that, almost surely, for every thickening $\mathcal{P}_{i}^{(k)}$ of an element $P_{i}$ of the cover, either $| \mathcal{P}_{i}^{(k)} \cap \mathcal{S}_{t} |$ is eventually greater than $N$ or  $| \mathcal{P}_{i}^{(k)} \cap \{ L_{t} \}_{t \in \mathbb{N}} |$ is finite. Note that it is trivial that either $| \mathcal{P}_{i}^{(k)} \cap \{L_{t} \}_{t \in \mathbb{N}} |$ is eventually greater than $N$ or  $| \mathcal{P}_{i}^{(k)} \cap \{ L_{t} \}_{t \in \mathbb{N}} |$ is finite; the goal is to check that if $\{L_{t} \}_{t \in \mathbb{N}}$ visits $\mathcal{P}_{i}$ infinitely often,  $| \mathcal{P}_{i}^{(k)} \cap \mathcal{S}_{t} |$ must eventually be greater than $N$. \par 

To see this, we introduce a representation of the random variables used in step 12 of Algorithm \ref{alg:algOverview}. Recall that in this step, $L_{t}$ is added to $\mathcal{S}_{t}$ with probability $\beta$, independently of the rest of the history of the walk. We will split up the sequence $B_{t}$ of Bernoulli$(\beta)$ random variables according to the covering as follows: for each element $\mathcal{P}_{i}^{(k)}$ of the covering, let $\{ B^{(i,k)}_{t} \}_{t \in \mathbb{N}}$ be an i.i.d. sequence of Bernoulli random variables with success probability $\beta$. At the $m$th time $L_{t}$ is in $\mathcal{P}_{i}^{(k)}$, we use $B^{(i,k)}_{m}$ as the indicator function in step 12 of Algorithm \ref{alg:algOverview}. This does not affect the distribution of the steps that the algorithm takes. \par 

By the Borel-Cantelli lemma, we have for each $i,k$ that $\P[B^{(i,k)}_{t} = 1, \textrm{infinitely often}]=1$. If $B^{(i,k)}_{t} = 1$ infinitely often, then $| \mathcal{P}_{i}^{(k)} \cap \{ L_{t} \}_{t \in \mathbb{N}} | = \infty$ implies that for all $M < \infty$, we have $|\mathcal{P}_{i}^{(k)} \cap \mathcal{S}_{t} | > M$ eventually. Let $\mathcal{C}_{i,k}$ be the event that $| \mathcal{P}_{i}^{(k)} \cap \mathcal{S}_{t} | > N$ eventually and let $\mathcal{D}_{i,k}$ be the event that $| \mathcal{P}_{i}^{(k)} \cap \{ L_{t} \}_{t \in \mathbb{N}} | = \infty$. Then this argument implies that 
\be 
\P[\mathcal{C}_{i,k} | \mathcal{D}_{i,k}] = 1.
\ee
Since there are only countably many sets $\mathcal{P}_{i}^{(k)}$, we have
\be \label{IneqSomeSets}
\P[\cap_{i,k} \left( \mathcal{C}_{i,k} \cup \mathcal{D}_{i,k}^{c} \right)] = 1.
\ee 

Thus, conditioned on the almost sure event $\cap_{i,k} \left( \mathcal{C}_{i,k} \cup \mathcal{D}_{i,k}^{c} \right)$, all sets $\mathcal{P}_{i}^{(k)}$ that $L_{t}$ visits infinitely often will also contribute points to $\mathcal{S}_{t}$ infinitely often. \par 
Let $k(i) = \min \{ k \, :  | \mathcal{P}_{i}^{(k)} \cap \{ L_{t} \}_{t \in \mathbb{N}} | = \infty  \}$ when that set is non-empty, and set $k(i) = \infty$ otherwise. Let $I = \{ i \, : \, k(i) < \infty \}$. Finally, set 
\be \label{EqFunnyTime}
\tau_{r_{1}, r_{2}} = \inf \{ t \, : \, \forall i \in I, \, | \mathcal{P}_{i}^{(k(i))} \cap \mathcal{S}_{t} | \geq N \}.
\ee
Since $|I|$ is finite, we have shown that, for all $N, \delta > 0$ and $0 < r_{2} < r_{1} < \infty$, $\P[\tau_{r_{1},r_{2}}<\infty] = 1$. Finally, we observe that for all $\delta = \delta(\mathcal{X},d)$ sufficiently small, all $N \geq N_{\defi}$ and all $k \geq \max_{i \in I} k(i)$, 
\be \label{EqImportantTimeOrder}
\tau_{\mathcal{X}}^{(k)} \leq \tau_{\frac{2}{3} \mathcal{X}, \frac{4}{3}  \mathcal{X}}. 
\ee 
This completes the proof.
\end{proof} 
\begin{remark}
We will eventually see that, in the notation of the proof of Lemma \ref{LemInfApprIgnoreZero}, $k(i) = 1$ for all $i$.
\end{remark}
Next, we show that the approximation  $p_{\mathcal{S}_{t}}(x)$ of the posterior used at time $t$ is close to $p_{\infty}(X_{t})$ when $\mathcal{S}_{t}$ is $(\infty,0,R)$-good and $\| X_{t} \|$ is sufficiently large: 

\begin{lemma} [Approximation at Infinity] \label{LemmAppAtInf}
For all $\epsilon > 0$ and $k \geq 2$, there exists a constant $\mathcal{X} = \mathcal{X}(\epsilon) > 0$ so that, if $R_{t} < (\| L_{t} \| - (k-1)\mathcal{X}) \textbf{1}_{ \| L_{t} \|> k \mathcal{X}}$ and the set $\{ q_{t}^{(1)}, \ldots, q_{t}^{(N)} \}$ is $\lambda$-poised, then
\be 
| \log( p_{\mathcal{S}_{t}}(L_{t})) - \log(p_{\infty}(L_{t})) | < \epsilon + \lambda (N+1) G.
\ee 
\end{lemma}

\begin{proof}
Fix $\epsilon > 0$. By \eqref{EqGaussEnv} in Assumption \ref{ass:GE}, there exists some $\mathcal{X} = \mathcal{X}(\epsilon)$ so that $\| x \| > \mathcal{X}$ implies 
\be \label{IneqLagAppr1}
| \log(p(x | \d)) - \log(p_{\infty}(x)) | < G + \frac{\epsilon}{(N+1) \lambda}.
\ee
We fix this constant $\mathcal{X}$ in the remainder of the proof. \par 

Denote by $\{ f_{i} \}_{i=1}^{N+1}$ the Lagrange polynomials associated with  the set $\{ q_{t}^{(1)}, \ldots, q_{t}^{(N)} \}$.  By Lemma 3.5 of \citep{Conn2009}, 
\be 
\left | \log(p_{\mathcal{S}_{t}}(L_{t})) - \log(p_{\infty}(L_{t})) \right | &= | \sum_{i} f_{i}(L_{t}) \log(p(q_{t}^{(i)}| \d)) - \log(p_{\infty}(L_{t}))  | \\
&\leq  | \sum_{i} \log(p_{\infty}(q_{t}^{(i)})) f_{i}(L_{t}) - \log(p_{\infty}(L_{t})) | \\
&\hspace{2cm}+ \sum_{i} | \log( p( q_{t}^{(i)} | \d )) - \log( p_{\infty}( q_{t}^{(i)} ) ) | \, | f_{i}(L_{t}) | \\
&\leq 0 + (N+1) \lambda \sup_{i} | \log( p( q_{t}^{(i)} | \d ) )- \log(p_{\infty}(q_{t}^{(i)})) |
\ee 
where the last line follows from the definition of Lagrange polynomials and Definition 4.7 of \citep{Conn2009}. Under the assumption $\| q_{t}^{(i)} \| > \mathcal{X}$ for $\| L_{t} \| - R_{t} > (k-1)\mathcal{X}$, the conclusion follows from inequality \eqref{IneqLagAppr1}. 
\end{proof}

For $\epsilon > 0$, define $V_{\epsilon}(x) = V(x)^{\frac{1}{1 + \epsilon}}$, where $V$ is defined in Equation \eqref{eqn:Vx}. Denote by $\alpha_{\infty}(x,y)$ the acceptance function of a Metropolis-Hastings chain with proposal kernel $L$ and target distribution $p_{\infty}$, and recall that $\tilde{\alpha}_{t}(x,y)$ as given in Equation \eqref{EqAlphDef} is the acceptance function for $\tilde{K}_{t}$. We show that $\tilde{K}_{t}$ inherits a drift condition from $K_{\infty}$:   

\begin{lemma} [Drift Condition] \label{LemmaInfDrift}
For $0 < \delta < \frac{1}{10}$ and $\mathcal{Y}, \mathcal{T} < \infty$, let $\mathcal{F}$ be the event that
\be \label{IneqTVDistAss}
| \tilde{\alpha}_{t}(X_{t},L_{t}) - \alpha_{\infty}(X_{t},L_{t}) | < \delta + 2 \textbf{1}_{|X_{t} | < \mathcal{Y}} + 2 \textbf{1}_{|L_{t} | < \mathcal{Y}}
\ee 
for all $t > \mathcal{T}$. Then, for $\epsilon = \epsilon_0$ as given in item 1 of Assumption \ref{assump:RoTw}, and all $\delta < \delta_{0}(\epsilon, a, b,V) < \frac{1}{10}$ sufficiently small and $\mathcal{Y}$ sufficiently large, $X_{t}$ satisfies a drift condition of the form: 
\be \label{IneqDriftCondition}
\E[V_{\epsilon}(X_{t+1}) \textbf{1}_{\mathcal{F}} | X_{t}, \mathcal{S}_{t}] \leq a_1 V_{\epsilon}(X_{t}) + b_1
\ee 
for some $0 \leq a_1 < 1$, $0 \leq b_1 < \infty$ and for all $ t > \mathcal{T}$. 
\end{lemma}

\begin{proof}
Assume WLOG that $\mathcal{T} = 0$. Let $Z_{t}$ be a Metropolis-Hastings Markov chain with proposal kernel $L$ and target distribution $p_{\infty}$. By Jensen's inequality and Assumption \ref{assump:RoTw}
\be
\E[V_{\epsilon}(Z_{t+1}) | Z_{t} = x] \leq a_{\epsilon} V_{\epsilon}(x) + b_{\epsilon}
\ee 
for some $0 < a_{\epsilon} < 1$ and some $0 \leq b_{\epsilon} < \infty$. \par 

Assume $X_{t} = x$ and fix $\delta$ so that $\delta < \delta_{0}$ and $(1 + 3 \delta)a_{\epsilon} < a_{\epsilon} + \frac{1}{2}( 1 - \alpha_{\epsilon})$. Then
\be 
\E[V_{\epsilon}(X_{t+1}) \textbf{1}_{\mathcal{F}}| X_{t} = x, \mathcal{S}_{t}] &\leq \int_{y \in \mathbb{R}^{d}} \left( \tilde{\alpha}_{t}(x,y) V_{\epsilon}(y) + (1 - \tilde{\alpha}_{t}(x,y)) V_{\epsilon}(x) \right) \ell(x,y)  dy \\
&\leq \int_{\mathbb{R}^{d} \backslash [-\mathcal{Y}, \mathcal{Y}]^{d}} \left( e^{2 \delta} \alpha_{\infty}(x,y) V_{\epsilon}(y) +  \left(1 - e^{-2 \delta}  \alpha_{\infty}(x,y) \right) V_{\epsilon}(x) \right) \ell(x,y) dy \\
&\hspace{2cm}+ \int_{y \in [-\mathcal{Y}, \mathcal{Y}]^{d}} \left( V_{\epsilon}(x) + \sup_{\| z \| \leq \mathcal{Y}} V_{\epsilon}(z) \right) \ell(x,y) dy \\
&\leq (1 + 3 \delta) \int_{\mathbb{R}^{d} } \left(  \alpha_{\infty}(x,y) V_{\epsilon}(y) +  \left(1 -  \alpha_{\infty}(x,y) \right) V_{\epsilon}(x) \right) \ell(x,y) dy \\
&\hspace{2cm}+ \left( V_{\epsilon}(x) + \sup_{\| z \| \leq \mathcal{Y}} V_{\epsilon}(z) \right) L \left( x, [-\mathcal{Y}, \mathcal{Y}]^{d} \right) \\
&\leq (1+ 3 \delta)  a_{\epsilon} V_{\epsilon}(x) + (1 + 3 \delta) b_{\epsilon} + \left( V_{\epsilon}(x) + \sup_{\| z \| \leq \mathcal{Y}} V_{\epsilon}(z) \right) L \left( x, [-\mathcal{Y}, \mathcal{Y}]^{d} \right).
\ee 
Since $\delta < \frac{1}{10}$ and $(1 + 3 \delta) a_{\epsilon} < a_{\epsilon} + \frac{1}{2}( 1 - \alpha_{\epsilon})$, we have 
\be 
\E[V_{\epsilon}(X_{t+1}) \textbf{1}_{\mathcal{F}} | X_{t} = x, \mathcal{S}_{t}] \leq (a_{\epsilon} + \frac{1}{2}( 1 - \alpha_{\epsilon}) ) V(x) + (1 + 3 \delta) b_{\epsilon} + \left( V_{\epsilon}(x) + \sup_{\| z \| \leq \mathcal{Y}} V_{\epsilon}(z) \right) L \left( x, [-\mathcal{Y}, \mathcal{Y}]^{d} \right).
\ee  
Since $V_{\epsilon}(x) L \left( x, [-\mathcal{Y}, \mathcal{Y}]^{d} \right)$ is uniformly bounded in $x$ for all fixed $\mathcal{Y}$ by item 2 of Assumption \ref{assump:RoTw}, the claim follows with
\be 
a_1 &= a_{\epsilon} + \frac{1}{2}( 1 - \alpha_{\epsilon}) < 1, \\
b_1 &= 2 b _{\epsilon} + \sup_{x} V_{\epsilon}(x) L \left( x, [-\mathcal{Y}, \mathcal{Y}]^{d}\right)  + \sup_{\| z \| \leq \mathcal{Y}} V_{\epsilon}(z), 
\ee 
finishing the proof.
\end{proof}

We use these bounds to show that some compact set is returned to infinitely often: 

\begin{lemma} [Infinitely Many Returns] \label{LemmaInfManyRets}
For $G < G(L, p_{\infty}, \lambda, N)$ sufficiently small, there exists a compact set $\mathcal{A}$ that satisfies $\P[\sum_{t \in \mathbb{N}} \textbf{1}_{X_{t} \in \mathcal{A}} = \infty] = 1$.
\end{lemma}
\begin{proof}
Combining Lemmas \ref{LemInfApprIgnoreZero}, \ref{LemmAppAtInf} and \ref{LemmaInfDrift}, there exists some number $\mathcal{X} > 0$ and almost surely finite random time $\tau_{\mathcal{X}}$  so that $X_{t}$ satisfies a drift condition of the form
\be 
\E[V(X_{t+1}) \textbf{1}_{t > \tau_{\mathcal{X}}} | X_{t} = x, \mathcal{S}_{t}] \leq aV(x) + b
\ee 
for some function $V$ and constants $0 \leq a < 1$, $b < \infty$
. The existence of a recurrent compact set follows immediately from this drift condition and Lemma 4 of \citep{Rose95}. \end{proof}

This allows us to slightly strengthen Lemma \ref{LemmaInfDrift}: 

\begin{lemma}  \label{CorInfDrift}
All times $\tau_{\mathcal{X}, 2 \mathcal{X}}$ of the form given in Equation \eqref{EqFunnyTime} satisfy $\P[\tau_{\mathcal{X}, 2 \mathcal{X}} < \infty] = 1$ and are stopping times with respect to $\{ \mathcal{S}_{t} \}$. Furthermore, for $G < G(L, p_{\infty}, \lambda, N)$ sufficiently small, there exists a random time $\tau$ of the form given in Equation \eqref{EqFunnyTime} so that
\be \label{IneqDriftCondition2}
\E[V_{\epsilon}(X_{t+1}) \textbf{1}_{\tau < t} | X_{t}, \mathcal{S}_{t}] \leq a_1 V_{\epsilon}(X_{t}) + b_1
\ee 
for some $0 \leq a_1 < 1$, $0 \leq b_1 < \infty$. 
\end{lemma}
\begin{proof}
By inequality \eqref{EqImportantTimeOrder}, there exists a random time $\tau \equiv \tau_{\mathcal{X}, 2 \mathcal{X}}$ of the form \eqref{EqFunnyTime} that is at least as large as the random time $\tau_{\mathcal{X}}$ constructed in the proof of Lemma \ref{LemmaInfManyRets} and that satisfies $\P[\tau < \infty] = 1$. As shown in Lemma \ref{LemmaInfManyRets}, an inequality of the form \eqref{IneqDriftCondition2} holds for $\tau_{\mathcal{X}}$, and so the same inequality must also hold with $\tau_{\mathcal{X}}$ replaced by the larger time $\tau \geq \tau_{\mathcal{X}}$.

The only detail to check is that all random times $\tau_{\mathcal{X}, 2 \mathcal{X}}$ of the form \eqref{EqFunnyTime} are stopping times with respect to $\{ \mathcal{S}_{t} \}_{t \in \mathbb{N}}$. Let $\{ P_{i} \}$ be the partition associated with $\tau_{\mathcal{X}, 2 \mathcal{X}}$, as constructed in Lemma \ref{LemInfApprIgnoreZero}. By Lemma \ref{LemmaInfManyRets} and part 3 of Assumption \ref{assump:RoTw}, we have $\P[|\{ L_{t} \}_{t \in \mathbb{N}} \cap \mathcal{P}_{i}^{(1)}| = \infty] = 1$ for all $i$. Thus, in the notation of Lemma \ref{LemInfApprIgnoreZero}, $I^{c} = \emptyset$ and $k(i) = 1$ for all $i \in I$. Thus, we have shown that $\tau_{\mathcal{X}, 2 \mathcal{X}} = \inf \{ t \, : \, \forall i, \, | \mathcal{P}_{i}^{(1)} \cap \mathcal{S}_{t} | \geq N \}$, which is clearly a stopping time with respect to $\{ \mathcal{S}_{t} \}_{t \in \mathbb{N}}$, and the proof is finished.
\end{proof}

We now finish our proof of Theorem \ref{ThmErgGaussEnv} analogously to our proof of Theorem \ref{ThmCompSup}.

The following bound is almost identical to Lemma \ref{LemGridRefComp}, but now proved under the Gaussian envelope assumption for the target density.

\begin{lemma} [Convergence of Kernels] \label{LemGridRefComp2}
Let the assumptions stated in the statement of  Theorem \ref{ThmErgGaussEnv} hold and fix a compact set $\mathcal{A} \subset \Theta$. For all $\delta > 0$, there exists a stopping time $\tau = \tau(\delta)$ with respect to $\{ \mathcal{S}_{t} \}_{t \in \mathbb{N}}$ so that
\be \label{IneqFirstGoal} 
\sup_{t > \tau} \, \sup_{x \in \mathcal{A}} \| K_{\infty}(x,\cdot) - \tilde{K}_{t}(x,\cdot) \|_{\TV} < \delta 
\ee 
and so that $\P[\tau < \infty] = 1$.
\end{lemma}

\begin{proof} 
Fix a constant $0 < R < \infty$ so that $\mathcal{A} \subset \mathcal{B}_{R}(0)$. By results in \citep{Conn2009}, for any $\lambda, \alpha > 0$, there exists a constant $c = c(\alpha, \lambda) > 0$ so that $\sup_{\theta \in \mathcal{A}}|p_{\mathcal{S}}(\theta) - p(\theta | \d) | < \alpha$ if $\mathcal{S}$ is $\lambda$-poised and $(c,R,R)$-good. Set $c = c(\epsilon, \lambda)$ and define $\tau' = \inf \{ t \, : \, \mathcal{S}_{t} \text{ is } (c,R,R)- \text{good} \}$. By definition, $\tau'$ is a stopping time with respect to $\{ \mathcal{S}_{t} \}_{t \in \mathbb{N}}$ that satisfies inequality \eqref{IneqFirstGoal}. We now check that $\P[\tau' < \infty] = 1$.  
By the assumption that $\ell(x,y)$ is bounded away from 0, there exist $\epsilon>0$ and measures $\mu$, $\{r_{x}\}_{x \in \Theta}$ so that 
\be \label{EqRepMin}
L(x, \cdot) = \epsilon \mu(\cdot) + (1 - \epsilon) r_{x}(\cdot).
\ee 

Let $\{A_{i} \}_{i \in \mathbb{N}}$ and $\{ B_{i} \}_{i \in \mathbb{N}}$ be two sequences of i.i.d. Bernoulli random variables, with success probabilities $\epsilon$ and $\beta$ respectively. Let $\tau_{0} = \inf \{ t \, : \, X_{t} \in \mathcal{A} \}$ and define inductively $\tau_{i+1} = \inf \{ t > \tau_{i} + 1 \, : \, X_{t} \in \mathcal{A} \}$. By equality \eqref{EqRepMin}, it is possible to couple the sequences $\{ X_{t} \}_{t \in \mathbb{N}}, \{A_{i} \}_{i \in \mathbb{N}}$ so that 
\be \label{EqDecompProp}
\P[L_{\tau_{i}} \in \cdot | \tau_{i}, X_{\tau_{i}}, A_{i} = 1] &=   \mu(\cdot)  \\
\P[L_{\tau_{i}} \in \cdot | \tau_{i}, X_{\tau_{i}}, A_{i} = 0] &=   r_{X_{\tau_{i}}}(\cdot).
\ee 

We can further couple $\{B_{i}\}_{i \in \mathbb{N}}$ to these sequences by using $B_{i}$ for the random variable in step 12 of Algorithm \ref{alg:algOverview} at time $\tau_{i}$. That is, when running  Algorithm \ref{alg:algOverview}, we would run the subroutine $\mathrm{RefineNear}$ in step 13 of the algorithm at time $t = \tau_{i}$ if $B_{i} = 1$, and we would not run that subroutine in that step at that time if $B_{i} = 0$.  Define $I = \{ i \in \mathbb{N} \, : \, A_{i} = B_{i} = 1 \}$. Under this coupling of $\{A_{i} \}_{i \in \mathbb{N}}, \{B_{i} \}_{i \in \mathbb{N}}$, and $\{X_{t} \}_{t \in \mathbb{N}}$,
\be 
\{ L_{\tau_{i}} \}_{i \in I, \, \tau_{i} < t }  \subset \mathcal{S}_{t}.
\ee 
Furthermore, $\{L_{\tau_{i}} \}_{i \in I, \, i \leq N }$ is an i.i.d sequence of $N$ draws from $\mu$, and by Lemma \ref{LemmaInfManyRets}, $\P[\tau_{i} < \infty] = 1$ for all $i$. Let $\mathcal{E}_{j}$ be the event that $\{ L_{\tau_{i}} \}_{i \leq j}$ is $(c,R,R)$-good. We have $\tau' \leq \inf \{ \tau_{j} \, : \, \mathcal{E}_{j} \text{ holds} \}$. By independence of the sequence $\{ L_{\tau_{i}} \}_{i \in \mathbb{N}}$, we obtain
\be 
\P[\tau' < \infty] \geq \liminf_{j \rightarrow \infty} \P[\mathcal{E}_{j}] = 1.
\ee 

This argument shows that, for any compact set $\mathcal{A}$, there exists a stopping time $\tau'$ with respect to $\{ \mathcal{S}_{t} \}_{t \in \mathbb{N}}$ so that $\P[\tau' < \infty] = 1$ and so that 
\be \label{IneqNextToLastLemmaConc}
\sup_{t > \tau'} \sup_{x \in \mathcal{A}} \| \tilde{K}_{t}(x,\cdot) - K_{\infty}(x,\cdot) \|_{\TV} < \delta.
\ee 
This completes the proof of the Lemma.

\end{proof}

We are finally ready to prove Theorem \ref{ThmErgGaussEnv}: 
\begin{proof} [Proof of Theorem \ref{ThmErgGaussEnv}]
As with the proof of Theorem \ref{ThmCompSup}, it is sufficient to show that, for all $\epsilon, \delta, G > 0$ sufficiently small and all $B \gg \epsilon^{-1}$ sufficiently large, the conditions of Theorem \ref{CorMainProbIneq} can be satisfied for some time $\mathcal{T} = \mathcal{T}_{\epsilon,\delta}$ with the same drift function $V$ and constants $\alpha, a,b$. 

By Assumption \ref{assump:RoTw} and Theorem 12 of \cite{Rose95}, inequality \eqref{IneqBasicResPos3} holds for some $\alpha > 0$ and $\{C_{x} \}_{x \in \Theta}$. For any fixed $0 < B < \infty$ and all $0 < G,\delta$ sufficiently small, Lemma \ref{LemGridRefComp2} implies that there exists some almost surely finite stopping time $\tau_{1} = \tau_{1}(\delta)$ so that inequality \eqref{IneqBasicResPos2} holds for the set $\mathcal{F}_{1} = \{ \tau_{1} > t \}$. Lemma \ref{CorInfDrift} implies that, for all $G > 0$ sufficiently small, there exists some almost surely finite stopping time $\tau_{2}$ so that inequality \eqref{IneqBasicResPos4} holds for the set $\mathcal{F}_{2} = \{ \tau_{2} > t \}$. Choose $\mathcal{T}$ to be the smallest integer so that $\P[\max(\tau_{1}, \tau_{2}) > \mathcal{T}] <  \epsilon$ and set $\mathcal{F} = \{ \min(\tau_{1}, \tau_{2}) > \mathcal{T} \}$. We then have that inequalities \eqref{IneqBasicResPos2} and \eqref{IneqBasicResPos4} are satisfied.  Finally, inequality \eqref{IneqBasicResPos1} holds by part 2 of Assumption \ref{assump:RoTw}. We have shown that there exist fixed values of $\mathcal{C}$ and $\alpha$ so that the conditions of Theorem \ref{CorMainProbIneq} hold for all $\epsilon, \delta > 0$ sufficiently small.  We conclude that, for all $\epsilon, \delta > 0$ sufficiently small, 
\be 
\limsup_{T \rightarrow \infty} \| \mathcal{L}(X_{T}) - \pi \|_{\TV} \leq 3\epsilon  + \delta \frac{ \log \left( \frac{e \delta}{\mathcal{C} \log(1-\alpha)} \right) }{\log(1 - \alpha)} +  \frac{4b}{a B} \lceil \frac{ \log \left( \frac{\delta}{\mathcal{C} \log(1-\alpha)} \right) }{\log(1 - \alpha)} + 1 \rceil.
\ee  
Letting $B$ go to infinity, then $\delta$ go to 0 and finally $\epsilon$ go to 0 completes the proof.
\end{proof}
\subsection{Alternative assumptions}
\label{sec:alternateAssumptions}
In this section, we briefly give other sufficient conditions for ergodicity. We do not give detailed proofs but highlight the instances at which our current arguments should be modified.\par
 The central difficulty in proving convergence of our algorithm is that, in general, the local polynomial fits we use may be very poor when $R_{t}$ is large. This difficulty manifests in the fact that, for most target distributions, making the set $\mathcal{S}$ a $(c,r,R)$-good set does not guarantee that $\tilde{K}_{\mathcal{S}}$ inherits a drift condition of the form \eqref{IneqBasicResPos4} from $K_{\infty}$, for any value of $c,r,R$. Indeed, no property that is monotone in the set $\mathcal{S}$ can guarantee that $\tilde{K}_{\mathcal{S}}$ satisfies a drift condition. In a forthcoming project focused on theoretical issues, we plan to show convergence based on drift conditions that only hold `on average' and over long time intervals. 
There are several other situations under which it is possible to guarantee the eventual existence of a drift condition, and thus ergodicity:

\begin{enumerate}
\item Fix a function $\delta_{0}: \Theta \rightarrow \mathbb{R}^{+}$ and add the step ``If $R_{t} > \delta_{0}(\theta^{+})$, $\mathcal{S} \leftarrow \{ (\theta^{+}, f(\theta^{+}) \} \cup \mathcal{S})$" between steps 7 and 8 of Algorithm \ref{alg:algOverview}. If $\lim_{r \rightarrow \infty} \sup_{\| x \| \geq r} \delta_{0}(x) = 0$ and 
\be
\lim_{r \rightarrow \infty} \sup_{\| x \| \geq r} \max( \| p'(\theta | \d) \|, \| p''(\theta | \d) \|) = 0,
\ee
then the main condition of Lemma \ref{LemmaInfDrift}, inequality \eqref{IneqTVDistAss} (with $\alpha_{\infty}$ replaced by the acceptance function of $K$), holds by a combination of Theorems 3.14 and 3.16 of \citep{Conn2009}. If $p( \theta | \d)$ has sub-Gaussian tails, the proof of Lemma \ref{LemmaInfDrift} can then continue largely as written if we replace $p_{\infty}(x)$ with $p(x | \d)$ wherever it appears. Since the Gaussian envelope condition is only used to prove that the condition in Lemma \ref{LemmaInfDrift} holds, Theorem \ref{ThmErgGaussEnv} holds with the Gaussian envelope condition replaced by these requirements.
\item Similar results sometimes hold if we only require that $\delta_{0}(x) \equiv \delta_{0}$ be a sufficiently small constant. Theorem 1 of \cite{FHL2013}, combined with Theorems 3.14 and 3.16 of \citep{Conn2009}, can be used to obtain weaker sufficient conditions under which the condition in Lemma \ref{LemmaInfDrift} holds. 
\item If $d=1$, $N_{\defi} = 2$, and the approximations in Algorithm \ref{alg:algOverview} are made using linear rather than quadratic models, we state without proof that a drift condition at infinity proved in Lemma \ref{LemmaInfDrift} can be verified directly. For $d \geq 2$, more work needs to be done.
\item Finally, we discuss analogous results that hold for other forms of local approximation, such as Gaussian processes. When the target distribution is compact, we expect Theorem \ref{ThmCompSup} to hold as stated whenever local approximations to a function based on $(c,R,R)$-good sets converge to the true function value as $c$ goes to 0. In our proof of Theorem \ref{ThmCompSup}, we cite \citep{Conn2009} for this fact. The proof of Theorem \ref{ThmCompSup} will hold as stated for other local approximations if all references to \citep{Conn2009} are replaced by references to appropriate analogous results. Such results typically hold for reasonably constructed local approximation strategies \citep{Cleveland1996, Atkeson1997}.

When the target distribution is not compact, modifying our arguments can be more difficult, though we expect similar conclusions to often hold.
\end{enumerate}
\subsection{Examples for parameter choices}
\begin{example}[Decay Rate for $\beta$] \label{ExDecRateBeta}
We note that if $\beta_{t}$ decays too quickly, our sampler may not converge, even if $\gamma_{t} \rightarrow 0$ at any rate. Consider the proposal distribution $L$ that draws i.i.d.\ uniform samples from $[0,1]^{d}$ and let $\lambda(\cdot)$ denote the Lebesgue measure. Consider a target distribution of the form $p(\theta | \d) \propto \textbf{1}_{\theta \in G}$ for set $G$ with Lebesgue measure $0 < \lambda(G) < 1$. If $\sum_{t} \beta_{t} < \infty$, then by Bo+rel-Cantelli, the probability $p = p \left( \{ \beta_{t} \}_{t \in \mathbb{N}} \right)$ that no points are added to $\mathcal{S}$ except during the initial choice of reference points or failed cross-validation checks is strictly greater than 0. With probability $ \lambda(G)^{k} > 0$, the first $k$ reference points are all in $G$. But if both these events happen, all cross-validation checks are passed for any $\gamma > 0$, and so the walk never converges; it samples from the measure $\lambda$ forever. 
\end{example}
\begin{example}[Decay Rate for $\gamma$]\label{ex:decrag}
We note that we have not used the assumption that $\gamma < \infty$ anywhere. As pointed out in Example \ref{ExDecRateBeta}, in a way this is justified---we can certainly find sequences $\{ \beta_{t} \}_{t \in \mathbb{N}}$ and walks that are not ergodic for any sequence $\gamma_{t} > 0$ converging to zero at any rate. \par 

In the other direction, there exist examples for which having any reasonable fixed value of $\gamma$ gives convergence, even with $\beta = 0$. We point out that this depends on the initially selected points; one could be unlucky and choose points with log-likelihoods that happen to lie exactly on some quadratic that does not match the true distribution. Consider a target density $\pi(x) \propto 1 + C \textbf{1}_{x > \frac{1}{2}}$ on $[0,1]$ with independent proposal moves from the uniform measure on $[0,1]$. To simplify the discussion, we assume that our approximation of the density at each point is linear and based exactly on the three nearest sampled points. Denote by $\mathcal{S}_{t}$ the points which have been evaluated by time $t$, and let $\mathcal{S}_{0} = \{ \frac{1}{8}, \frac{2}{8}, \frac{3}{8}, \frac{5}{8}, \frac{6}{8}, \frac{7}{8} \}$. Write $x_{1}, \ldots, x_{m(t)} = \mathcal{S}_{t} \cap [0, \frac{1}{2}]$ and $x_{m(t)+1}, \ldots, x_{n(t)} = \mathcal{S}_{t} \cap [\frac{1}{2},1]$. It is easy to check that 
\be \label{IneqEasyTVBound}
\| \mathcal{L}(X_{t+1}) - \pi \|_{\TV} \leq x_{m(t) + 3} - x_{m(t) - 2}.
\ee
It is also easy to see that with probability one, for any $\gamma < \frac{1}{2}$, there will always be a subinterval of $[x_{m(t)-2}, x_{m(t)+3}]$ with strictly positive measure for which a cross-validation check will fail. Combining this with inequality \eqref{IneqEasyTVBound} implies that the algorithm will converge in this situation, even with $\beta = 0$. Furthermore, in this situation choosing $\beta \equiv 0$ results in a set $\mathcal{S}_{t}$ that grows extremely slowly in $t$, without substantially increasing bias. \par 
\end{example}

\section{Genetic toggle switch inference problem}
\label{apx:genetics}

Here we provide additional details about the setup of the genetic toggle switch inference problem from Section~\ref{s:genetics}. This genetic circuit has a bistable response to the concentration of an input chemical, [IPTG]. Figure \ref{fig:GeneticsResponse} illustrates these high and low responses, where the vertical axis corresponds to the expression level of a particular gene. \citep{Gardner2000} proposed the following differential-algebraic model for the switch:
\begin{eqnarray}
\frac{du}{dt} &= &\frac{\alpha_1}{1 + v^\beta} - u,\\
\frac{dv}{dt} &= &\frac{\alpha_2}{1+w^\gamma}-v, \nonumber \\
w &= & \frac{u}{(1+\mathrm{[IPTG]}/K)^\eta}. \nonumber
\label{e:genesystem}
\end{eqnarray}
The model contains six unknown parameters $Z_\theta = \{\alpha_1, \alpha_2, \beta, \gamma, K, \eta\} \in \mathbb{R}^6$, while the data correspond to observations of the steady-state values $v(t=\infty)$ for six different input concentrations of  [IPTG], averaged over several trials each. As in \citep{Marzouk2009a}, the parameters are centered and scaled around their nominal values so that they can be endowed with uniform priors over the hypercube $[-1,1]^6$. Specifically, the six parameters of interest are normalized around their nominal values to have the form
\begin{equation*}
Z_i = \bar{\theta}_i(1+\zeta_i \theta_i), \ i=1, \ldots, 6,
\end{equation*}
so that each $\theta_i$ has prior $\text{Uniform}[-1,1]$. The values of $\bar{\theta}_i$ and $\zeta_i$ are given in Table \ref{tab:geneticParams}. The data are observed at six different values of [IPTG]; the first corresponds to the ``low'' state of the switch while the rest are in the ``high'' state. Multiple experimental observations are averaged without affecting the posterior by correspondingly lowering the noise; hence, the data comprise one observation of $v/v_\text{ref}$ at each concentration, where $v_\text{ref} = 15.5990$. The data are modeled as having independent Gaussian errors, \emph{i.e.}, as draws from $\mathcal{N}(d_i, \sigma_i^2)$, where the high- and low-state observations have different standard deviations, specified in Table \ref{tab:geneticData}. The forward model may be computed by integrating the ODE system (\ref{e:genesystem}), or more simply by iterating until a fixed point for $v$ is found. 

\begin{figure}[htbp]

\centering
\includegraphics{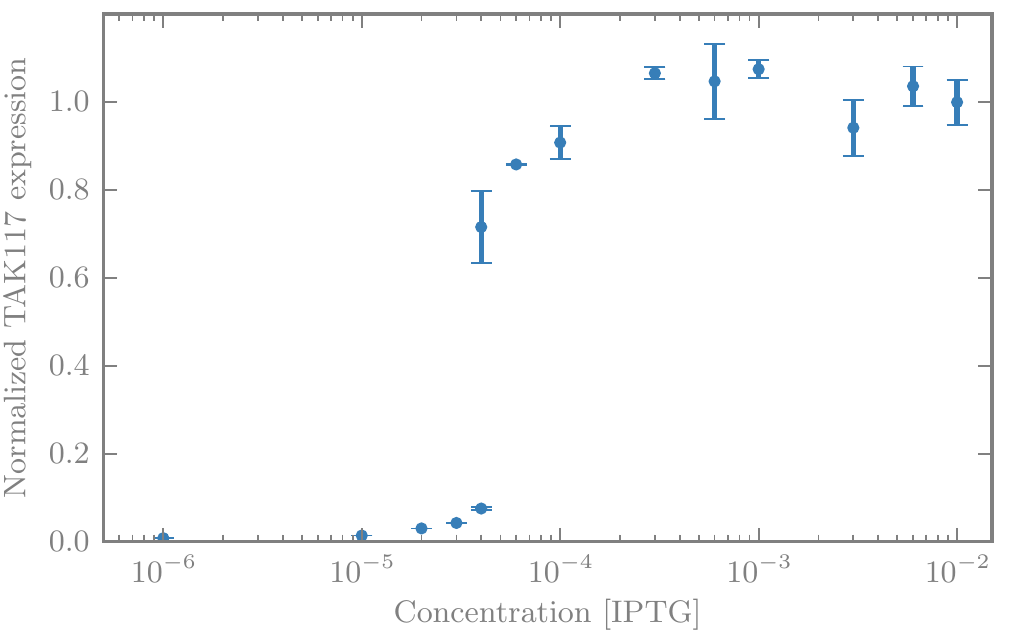}
\caption{Response of the pTAK117 genetic toggle switch to the input concentration of IPTG \citep{Gardner2000}. The plot shows the mean and standard deviation of the experimentally-observed gene expression levels over a range of input concentrations. Expression levels are normalized by the mean response at the largest IPTG concentration.}
\label{fig:GeneticsResponse}
\end{figure}

\begin{table}
\caption{\label{tab:geneticParams}Normalization of the parameters in the genetic toggle switch example.}

\ra{1.2}
\centering
\begin{tabular}{l l l l l l l}
\toprule
& $\alpha_1$ & $\alpha_2$ & $\beta$ & $\gamma$ &  $K$ & $\eta$\\
\midrule
$\bar{\theta}_i$ & 156.25 & 15.6 & 2.5 & 1 & 2.0015 & 2.9618e-5\\
$\zeta_i$ & 0.20 & 0.15 & 0.15 & 0.15 & 0.30 & 0.2\\
\bottomrule
\end{tabular}
\end{table}

\begin{table}
\caption{\label{tab:geneticData}
Data and obervation error variances for the likelihood of the genetic toggle switch example.}
\ra{1.2}
\centering
\begin{tabular}{l l l l l l l}
\toprule
$[$IPTG$]$ & 156.25 & 15.6 & 2.5 & 1 & 2.0015 & 2.9618e-5\\
\midrule
$d_i$ &  0.00798491 & 1.07691684 & 1.05514201 & 0.95429837 & 1.02147051 & 1.0\\
$\sigma_i$ & 4.0e-5 & 0.005 & 0.005 & 0.005 & 0.005 & 0.005\\
\bottomrule
\end{tabular}
\end{table}

\section{Elliptic PDE inverse problem}
\label{apx:pde}

Here we provide details about the elliptic PDE inference problem. The forward model is given by the solution of an elliptic PDE in two spatial dimensions
\be
\nabla_{\mathbf{s}} \cdot \left ( k(\mathbf{s}, \theta) \nabla_{\mathbf{s}} u(\mathbf{s}, \theta) \right ) = 0 ,
\label{e:elliptic}
\ee
where $\mathbf{s} =(s_1, s_2) \in [0,1]^2$ is the spatial coordinate. The boundary conditions are
\begin{eqnarray*}
u(\mathbf{s}, \theta) |_{s_2 =0} &=& s_1,\\
u(\mathbf{s}, \theta) |_{s_2 =1} &=& 1-s_1,\\
\left. \frac{\partial u(\mathbf{s}, \theta)} {\partial s_1 } \right|_{s_1=0} &=& 0,\\
\left. \frac{ \partial u(\mathbf{s}, \theta)} {\partial s_1} \right |_{s_1=1} &=& 0.
\end{eqnarray*}
This PDE serves as a simple model of steady-state flow in aquifers and other subsurface systems; $k$ can represent the permeability of a porous medium while $u$ represents the hydraulic head. Our numerical solution of (\ref{e:elliptic}) uses the standard continuous Galerkin finite element method with bilinear basis functions on a uniform $30$-by-$30$ quadrilateral mesh. 

The log-diffusivity field $\log k(\mathbf{s})$ 
%
%
is endowed with a Gaussian process prior, with mean zero and an isotropic squared-exponential covariance kernel:
\begin{equation*}
C(\mathbf{s}_1, \mathbf{s}_2) = \sigma^2 \exp  \left ( - \frac{ \left \| \mathbf{s}_1 - \mathbf{s}_2 \right \|^2}{2 \ell^2} \right ),
\end{equation*}
for which we choose variance $\sigma^2 = 1$ and a length scale $\ell=0.2$. This prior allows the field to be easily parameterized with a Karhunen-Lo\`{e}ve (K-L) expansion \citep{Adler1981}:
\begin{equation*}
k(\mathbf{s}, \theta) \approx \exp \left( \sum_{i=1}^d \theta_i \sqrt{\lambda_i} k_i(\mathbf{s}) \right),
\end{equation*}
where $\lambda_i$ and $k_i(\mathbf{s})$ are the eigenvalues and eigenfunctions, respectively, of the integral operator on $[0,1]^2$ defined by the kernel $C$, and the parameters $\theta_i$ are endowed with independent standard normal priors, $\theta_i \sim \mathcal{N}(0,1)$. These parameters then become the targets of inference. In particular, we truncate the Karhunen-Lo\`{e}ve expansion at $d=6$ modes and condition the corresponding mode weights $(\theta_1, \ldots, \theta_6)$ on data. Data arise from observations of the solution field on a uniform $11\times 11$ grid covering the unit square. The observational errors are taken to be additive and Gaussian:
\begin{equation*}
d_j = u(\mathbf{s}_j, \theta) + \epsilon_j,
\end{equation*}
with $\epsilon_j \sim \mathcal{N}(0, 0.1^2)$. 

\end{document}